\newif\ifisarxiv
\isarxivtrue

\ifisarxiv
\documentclass[11pt]{article}
\usepackage{fullpage}
\usepackage{graphicx, graphics, epsfig, color}
\usepackage{booktabs}
\usepackage{subcaption}
\usepackage{amsmath,amssymb,amsthm}
\usepackage{tikz, pgfplots}
\usepackage{hyperref}
\hypersetup{colorlinks=false}
\hypersetup{
	colorlinks,
	linkcolor={red!40!gray},
	citecolor={blue!40!gray},
	urlcolor={blue!70!gray}
}
\newtheorem{theorem}{Theorem}
\newtheorem{proposition}[theorem]{Proposition}
\newtheorem{definition}[theorem]{Definition}
\newtheorem{remark}[theorem]{Remark}
\newtheorem{corollary}[theorem]{Corollary}
\newtheorem{lemma}[theorem]{Lemma}
\else
\documentclass[final,12pt]{colt2021} 
\fi

\def\nnz{{\mathrm{nnz}}}
\def\binom{{\mathrm{Binomial}}}

\def\less{{LESS}}
\def\epsinv{{\epsilon}}
\def\epssub{{\eta}}

\def\Ec{\mathcal E}

\newif\ifDRAFT
\DRAFTtrue
\ifDRAFT
\newcommand{\marrow}{\marginpar[\hfill$\longrightarrow$]{$\longleftarrow$}}
\newcommand{\niceremark}[3]
   {\textcolor{red}{\textsc{#1 #2:} \marrow\textsf{#3}}}
\newcommand{\ken}[2][says]{\niceremark{Ken}{#1}{#2}}

\newcommand{\michael}[2][says]{\niceremark{Michael}{#1}{#2}}
\newcommand{\michal}[2][says]{\niceremark{Michal}{#1}{#2}}
\newcommand{\feynman}[2][says]{\niceremark{Feynman}{#1}{#2}}
\else
\newcommand{\ken}[1]{}
\newcommand{\michael}[1]{}
\newcommand{\michal}[1]{}
\newcommand{\feynman}[1]{}
\fi

\def\ee{\mathrm{e}}

\def\xib{{\boldsymbol\xi}}

\def\Hbh{\widehat{\H}}

\def\S{\mathbf{S}}
\def\T{\mathbf{T}}

\def\s {\mathbf{s}}

\def\g{{\mathbf{g}}}

\def\L{\mathbf{L}}

\def\H{\mathbf H}

\def\Q{\mathbf Q}

\ifx\BlackBox\undefined
\newcommand{\BlackBox}{\rule{1.5ex}{1.5ex}}  
\fi
\DeclareMathOperator*{\argmin}{\mathop{\mathrm{argmin}}}

\DeclareMathOperator*{\diag}{\mathop{\mathrm{diag}}}
\def\x{\mathbf x}
\def\y{\mathbf y}

\def\z{\mathbf z}
\def\a{\mathbf a}
\def\b{\mathbf b}

\def\v{\mathbf v}

\def\e{\mathbf e}

\def\one{\mathbf 1}
\def\u{\mathbf u}

\def\X{\mathbf X}

\def\B{\mathbf B}
\def\A{\mathbf A}
\def\C{\mathbf C}
\def\U{\mathbf U}

\def\F{\mathbf F}
\def\D{\mathbf D}
\def\mG{\mathbf \Gamma}
\def\O{\mathbf O}
\def\V{\mathbf V}
\def\bbr{\mathbf r}
\def\bbR{\mathbf R}
\def\M{\mathbf M}

\def\Z{\mathbf Z}

\def\I{\mathbf I}

\def\A{\mathbf A}
\def\P{\mathbf P}
\def\J{\mathbf J}
\def\bR{\mathbf R}

\def\bS{\mathbf{\Sigma}}
\def\bhS{\mathbf{\hat\Sigma}}
\def\bLambda{\mathbf{\Lambda}}

\def\Ot{\widetilde{O}}

\def\E{\mathbb E}
\def\R{\mathbb R}

\def\Pr{\mathrm{Pr}}
\def\tr{\mathrm{tr}}

\def\Var{\mathrm{Var}}

\let\origtop\top
\renewcommand\top{{\scriptscriptstyle{\origtop}}} 

\definecolor{silver}{cmyk}{0,0,0,0.3}
\definecolor{yellow}{cmyk}{0,0,0.9,0.0}
\definecolor{reddishyellow}{cmyk}{0,0.22,1.0,0.0}
\definecolor{black}{cmyk}{0,0,0.0,1.0}
\definecolor{darkYellow}{cmyk}{0.2,0.4,1.0,0}
\definecolor{orange}{cmyk}{0.0,0.7,0.9,0}
\definecolor{darkSilver}{cmyk}{0,0,0,0.1}
\definecolor{grey}{cmyk}{0,0,0,0.5}
\definecolor{darkgreen}{cmyk}{0.6,0,0.8,0}

\ifx\proof\undefined
\newenvironment{proof}{\par\noindent{\bf Proof\ }}{\hfill\BlackBox\\[2mm]}
\fi

\ifx\theorem\undefined
\newtheorem{theorem}{Theorem}
\fi

\ifx\example\undefined
\newtheorem{example}{Example}
\fi

\ifx\condition\undefined
\newtheorem{condition}{Condition}
\fi
\ifx\property\undefined
\newtheorem{property}{Property}
\fi

\ifx\lemma\undefined
\newtheorem{lemma}{Lemma}
\fi

\ifx\proposition\undefined
\newtheorem{proposition}{Proposition}
\fi

\ifx\remark\undefined
\newtheorem{remark}{Remark}
\fi

\ifx\corollary\undefined
\newtheorem{corollary}{Corollary}
\fi

\ifx\definition\undefined
\newtheorem{definition}{Definition}
\fi

\ifx\conjecture\undefined
\newtheorem{conjecture}{Conjecture}
\fi

\ifx\axiom\undefined
\newtheorem{axiom}{Axiom}
\fi

\ifx\claim\undefined
\newtheorem{claim}{Claim}
\fi

\ifx\assumption\undefined
\newtheorem{assumption}{Assumption}
\fi

\ifx\condition\undefined
\newtheorem{condition}{Condition}
\fi

\ifx\example\undefined
\newtheorem{example}{Example}
\fi

\title{Sparse sketches with small inversion bias}
\usepackage{times}

\ifisarxiv
\author{
  Micha{\l} Derezi\'nski\thanks{Department of Statistics, University
    of California, Berkeley (\texttt{mderezin@berkeley.edu})}
  \and
  Zhenyu Liao\thanks{ICSI and Department of Statistics,  University
    of California, Berkeley (\texttt{zhenyu.liao@berkeley.edu})}
  \and
  Edgar Dobriban\thanks{Department of Statistics and Data Science,
    University of Pennsylvania (\texttt{dobriban@wharton.upenn.edu})}
  \and
  Michael W. Mahoney\thanks{ICSI and Department of Statistics,
  University of California, Berkeley (\texttt{mmahoney@stat.berkeley.edu})}
}
\else
\coltauthor{%
 \Name{Micha{\l} Derezi\'nski} \Email{mderezin@berkeley.edu}\\
 \addr Department of Statistics, UC Berkeley%
 \AND
 \Name{Zhenyu Liao} \Email{zhenyu.liao@berkeley.edu}\\
 \addr ICSI and Department of Statistics, UC Berkeley%
 \AND
 \Name{Edgar Dobriban} \Email{dobriban@wharton.upenn.edu}\\
 \addr Department of Statistics and Data Science, University of Pennsylvania%
 \AND
  \Name{Michael W. Mahoney} \Email{mmahoney@stat.berkeley.edu}\\
 \addr ICSI and Department of Statistics, UC Berkeley%
}
\fi

\begin{document}

\maketitle

\begin{abstract}%
For a tall $n\times d$ matrix $\A$ and a random $m\times n$
sketching matrix $\S$, the sketched estimate of the inverse covariance
matrix $(\A^\top\A)^{-1}$ is typically biased:
$\E[(\tilde\A^\top\tilde\A)^{-1}]\neq (\A^\top\A)^{-1}$, where
$\tilde\A=\S\A$. 
This phenomenon, which we call \emph{inversion bias}, arises, e.g., in statistics and distributed optimization, when averaging multiple independently constructed estimates of quantities that depend on the inverse covariance matrix. 
We develop a framework for analyzing inversion bias, based on our proposed
concept of an $(\epsilon,\delta)$-unbiased estimator for random
matrices. We show that when the sketching matrix $\S$ is dense and has i.i.d.\
sub-gaussian entries, then after simple rescaling, the estimator $(\frac
m{m-d}\tilde\A^\top\tilde\A)^{-1}$ is $(\epsilon,\delta)$-unbiased
for $(\A^\top\A)^{-1}$ with a sketch of size $m=O(d+\sqrt
d/\epsilon)$. In particular, this implies 
that for $m=O(d)$, the inversion bias of this estimator is $O(1/\sqrt
d)$, which is much smaller than the $\Theta(1)$ approximation error obtained
as a consequence of the subspace embedding guarantee for
sub-gaussian sketches. We then propose a new 
sketching technique, called LEverage Score Sparsified (\less)
embeddings, which uses ideas from both data-oblivious sparse
embeddings as well as data-aware leverage-based row sampling methods, to get $\epsilon$
inversion bias for sketch size $m=O(d\log d+\sqrt d/\epsilon)$ in time
$O(\nnz(\A)\log n + md^2)$, where $\nnz$ is the number of non-zeros. 
The key techniques enabling our analysis include an extension of a
classical inequality of Bai and Silverstein for random quadratic
forms, which we call the \emph{Restricted Bai-Silverstein~inequality};
and anti-concentration of the Binomial distribution via the
Paley-Zygmund inequality, which we use to prove a lower bound showing
that leverage score sampling sketches generally do not achieve small inversion
bias. 
\end{abstract}
\ifisarxiv\else
\fi
\section{Introduction}

Sketching has been widely used in the design of scalable algorithms,
perhaps most prominently in Randomized Numerical Linear Algebra
(RandNLA) due to applications in machine learning and data analysis.
In this approach, one randomly samples or computes a
random projection of the data matrix to construct a smaller matrix,
the sketch. One then uses the sketch as a surrogate to approximate
quantities of interest.  The analysis of these methods typically
proceeds via a Johnson-Lindenstrauss-type argument to establish that
the geometry of the matrix is not perturbed too much under the sketching
operation. These methods have yielded state-of-the-art in
worst-case analysis, high-quality numerical implementations, and
numerous applications in machine learning%
~\cite{Mah-mat-rev_JRNL,tropp2011structure,woodruff2014sketching,DM16_CACM,RandNLA_PCMIchapter_chapter}. 

In many cases, either to preserve the structure of the data or for algorithmic
reasons, one is interested in sparse sketches, i.e., random
transformations that are represented by matrices with most entries
exactly equal to zero. One class of sparse sketches includes row sampling
techniques, such as leverage score
sampling~\cite{Drineas08CUR,fast-leverage-scores,MMY15}, which are
typically \emph{data-aware}, in the sense that the sampling distribution depends on 
the given data matrix. Another important class of
methods uses \emph{data-oblivious} sparse embeddings, such as
the CountSketch~\cite{charikar2002finding,cw-sparse,mm-sparse,nn-sparse}, to construct 
sketches in time depending on the number of non-zeros (nnz) in the
input. 

In all these cases, one can show that the sketch will be an
approximation of the solution with high probability. However,
comparatively little is known about the average performance of these
sketches. 
In particular, there may be a systematic bias away from the solution,
which is problematic in many situations in statistics, machine
learning, and data analysis. Perhaps the most ubiquitous example of
this phenomenon is the systematic bias caused by matrix inversion, a
key component of algorithms in the aforementioned domains.
In this paper, we introduce the fundamental notion of inversion bias, which
provides a finer control over the sketched estimates involving
matrix inversion. We show that one can conveniently make the inversion bias
small with dense Gaussian and sub-gaussian sketches. 
We also show that some sparse sketches do not have this desired property.
Then, we provide a non-trivial new construction and algorithm,
using ideas from both data-oblivious projections and data-aware
sampling, to get small inversion bias even for very sparse sketches.

\subsection{Overview}
\label{s:overview}

Consider an $n\times d$ data matrix $\A$ of rank $d$, where $n\geq d$. 
In many applications, we wish to approximate quantities of
the form $F((\A^\top\A)^{-1})$, where $(\A^\top\A)^{-1}$ is the
$d\times d$ inverse data covariance and $F(\cdot)$ is a linear
functional. Our goal is to provide a finer control over the effect of
matrix inversion on the quality of such estimates.
Here are some of the motivating examples: 
\begin{itemize}
  \item The vector $(\A^\top\A)^{-1}\b$ is
the solution of ordinary least squares (OLS) when $\b=\A^\top\y$ for a
vector $\y$, arguably the most widely used multivariate statistical method
\cite{anderson1958introduction,rao1973linear,friedman2009elements},
and it is also crucial for the Newton's method in numerical optimization
\cite{boyd2004convex,nocedal2006numerical}. In particular, accurate
approximations of this vector lead directly to improved convergence guarantees for
many optimization algorithms \cite{pilanci2016iterative,distributed-newton}.
\item The scalar $\x^\top(\A^\top\A)^{-1}\x$ for a vector $\x$, has
  numerous use-cases: When $\x=\a_i$ is one of the rows of $\A$, then it
  represents the statistical leverage scores
  \cite{drineas2006sampling}; If $\x=\e_i$ is a standard basis vector,
 then this is the squared length of the confidence interval for the
 $i$-th coefficient in OLS
 \cite{anderson1958introduction,friedman2009elements}.
\item The scalar $\tr\,\C(\A^\top\A)^{-1}$ for a matrix $\C$, is used to
quantify uncertainty in statistical results, e.g., via the mean
squared error (MSE) of estimating the regression coefficients in OLS
\cite{anderson1958introduction,friedman2009elements}, and to
formulate widely used criteria from experimental design, e.g.,
A-designs and V-designs \cite{pukelsheim2006optimal,cox2000theory}.
\end{itemize}
More generally, 
  our work is also motivated by the important problem of inverse
  covariance estimation in statistics, machine learning, finance,
  signal processing, and related areas
  \cite{dempster1972covariance,meinshausen2006high,friedman2008sparse,lam2009sparsistency,cai2011constrained,ledoit2012nonlinear,marjanovic2015l_}. 
  In this area, we wish to estimate statistically the
  inverse covariance matrix of a population, or some of its
  functionals, based on a 
  finite number of samples.
  Furthermore, inverting covariance matrices
  occurs in Bayesian statistics
  \cite{hartigan1969linear,gelman2013bayesian}, Gaussian processes
  \cite{rasmussen2003gaussian}, as well as time series analysis and control,
  e.g., via the Kalman filter
  \cite{welch1995introduction,brockwell2009time}. 

When $n$ and $d$ are large, and particularly
when $n\gg d$, then the costs of storing the matrix~$\A$ and of
computing $(\A^\top\A)^{-1}$ are prohibitively large.
Matrix sketching has proven successful at drastically reducing these
costs by approximating the inverse covariance with a sketched estimate
$(\tilde\A^\top\tilde\A)^{-1}$ based on a smaller matrix
$\tilde\A=\S\A$, where $\S$ is a random $m\times n$ matrix and $m\ll
n$ \cite{Mah-mat-rev_JRNL,tropp2011structure,woodruff2014sketching,DM16_CACM,RandNLA_PCMIchapter_chapter}. 
As a concrete algorithmic motivation for our work, consider the following popular strategy for
boosting the quality of such estimates:
Construct multiple copies in parallel, based on independent sketches, 
and then average the estimates. This strategy is especially useful in distributed architectures,
where storage and computing resources are 
spread out across many
machines, and has commonly appeared in the literature~\cite{konecny-federated16a,konecny-federated16b,distributed-newton,determinantal-averaging}. While promising in practice, this 
averaging technique is fundamentally limited by the \emph{inversion bias}: even
though the sketched covariance estimate is unbiased,
$\E[\tilde\A^\top\tilde\A]=\A^\top\A$, its inverse in general is not unbiased, i.e.,
$\E[(\tilde\A^\top\tilde\A)^{-1}] \neq (\A^\top\A)^{-1}$.
When the sketch size $m$ is not much larger than the dimension $d$, the 
size of this bias can be very significant, even as large as the
approximation error, in which case averaging becomes ineffective. Motivated by this, we ask:
\begin{quote}
  When is the inversion bias small, relative to the
  approximation error?
\end{quote}
In this paper, we develop a framework for analyzing the inversion bias
of sketching, via the notion of an $(\epsilon,\delta)$-unbiased
estimator (Definition \ref{d:unbiased-estimator}), and
we show how it can be used to provide improved 
approximation guarantees for averaging. Through this framework, we
provide several contributions towards addressing the above question.

\textbf{Sub-gaussian sketches have small inversion bias.}
Arguably the most classical family of sketches consists of dense random
matrices $\S$ with i.i.d.\ sub-gaussian entries. These sketches offer strong 
relative error approximation guarantees via the so-called \emph{subspace embedding}
property, at the expense of high computational cost 
of the matrix product $\tilde\A=\S\A$. We show that, upon a simple
correction, sub-gaussian sketches 
are nearly-unbiased, i.e., their inversion bias is much smaller than the approximation error, which
means that averaging can be used to significantly improve the
approximation quality.
In particular, we show that, after a simple scalar
rescaling, the inverse covariance 
estimator of the form $(\frac{m}{m-d}\tilde\A^\top\tilde\A)^{-1}$
achieves $\epsilon$ inversion bias relative to $(\A^\top\A)^{-1}$
with a sketch of size only $m=O(d+ \sqrt d/\epsilon)$
(Proposition~\ref{t:subgaussian}).
In contrast, to ensure that
$(\frac{m}{m-d}\tilde\A^\top\tilde\A)^{-1}$ is an
$\epssub$ relative error approximation of $(\A^\top\A)^{-1}$ via 
the subspace embedding property, we need a sub-gaussian sketch of
size $m=\Theta(d/\epssub^2)$, which is comparatively larger if we let $\eta
= \Theta(\epsilon)$. This implies that  
an aggregate estimator obtained via averaging can with high probability produce
a relative error approximation that is by a factor of $O(1/\sqrt m)$
better than the approximation error offered by any one of the
estimators being averaged.  

\textbf{LEverage Score Sparsified (\less) embeddings.}
We show that existing algorithmically efficient sketching techniques
may not provide guarantees for the inversion bias that match those satisfied
by dense sub-gaussian sketches (see Theorem \ref{t:lower} for a lower bound
on leverage score sampling, and a discussion of other methods in
Appendix \ref{s:lower-bai-silverstein}). To address this, we propose a
new family of sketching methods, called LEverage 
Score Sparsified (\less) embeddings, which combines a data-oblivious sparsification
strategy reminiscent of the CountSketch with the data-aware
approach of approximate leverage score sampling. \less\ embeddings
have time complexity $O(\nnz(\A)\log n + md^2)$ and achieve
$\epsilon$ inversion bias with the sketch of size  $m=O(d\log d + \sqrt
d/\epsilon)$, nearly matching our guarantee for sub-gaussian sketches
(Theorem \ref{t:lsse}). Thus, our new algorithm provides a promising way to address
the fundamental problem of inversion bias, and it may have many other
applications in the future. Finally, our analysis reveals two structural
conditions for small inversion bias (Theorem~\ref{t:structural}), one
of which (Condition~\ref{cond2}, called the Restricted Bai-Silverstein condition) leads
to a generalization of a classical inequality used in random matrix
theory, and should be of independent~interest.

\subsection{Related work}
\label{s:related-work}

\paragraph{Distributed averaging.}
Averaging strategies have been studied extensively in the literature,
particularly in the context of machine learning and numerical
optimization. This line of work has proven particularly effective for
\emph{federated learning}
\cite{konecny-federated16a,konecny-federated16b}, where local storage
and communication 
bandwidth are particularly constrained. The performance of averaged
estimates was analyzed in numerous statistical learning settings
\cite{mcdonald09,mcdonald10,zhang-duchi-wainwright13,dobriban2018understanding,JMLR:v21:19-277} and in
stochastic first-order optimization \cite{parallel-sgd,agarwal-duchi11}. Of 
particular relevance to our results is a recent line of works on distributed
second-order optimization \cite{dane,disco,aide,distributed-newton}, as well as large-scale second-order optimization \cite{YGKM19_pyhessian_TR,YGSKM20_adahessian_TR},
since sketching is used there to 
estimate (implicitly) the inverse Hessian matrix which arises in
Newton-type methods. In particular,
\cite{determinantal-averaging,debiasing-second-order} pointed to
Hessian inversion bias as a key challenge in these approaches. 
To address it, their algorithms use non-i.i.d. sampling sketches based on Determinantal Point Processes%
~(DPPs) \cite{DM21_NoticesAMS}. DPP-based sketches are known to
correct inversion bias exactly 
\cite{unbiased-estimates-journal,correcting-bias-journal,minimax-experimental-design}.
However, state-of-the-art DPP sampling algorithms
\cite{leveraged-volume-sampling,dpp-sublinear,alpha-dpp} 
have time complexity $O(\nnz(\A)\log n + d^4\log d)$,
which is considerably more expensive than fast sketching techniques
when dimension $d$ is large.

\paragraph{Random matrix theory.}
When considering $\S \in \R^{m \times n}$ having i.i.d.\@ zero-mean rows,
$\A^\top \S^\top \S \A$ can be viewed as the popular \emph{sample covariance
  estimator} of the ``population covariance matrix'' $\A^\top \A \in \R^{d \times
  d}$. In this area, one often considers the matrix $(\A^\top \S^\top \S \A - z \I)^{-1}$ for $z \in
\mathbb C \setminus \R_+$, the so-called \emph{resolvent} matrix,
which plays a fundamental role in the literature of random matrix
theory (RMT)
\cite{marchenko1967distribution,bai2010spectral,edelman2005random,anderson2010introduction,couillet2011random,tao2012topics,bun2017cleaning}
and which is directly connected to the popular Marchenko-Pastur law
\cite{marchenko1967distribution}. The RMT literature focuses on the Stieltjes transform (that
is, the normalized trace of the resolvent) to investigate the limiting
eigenvalue distribution of large random matrices of the form $\A^\top
\S^\top \S \A$ as $m,n,d \to \infty$ at the same rate.
Here, we provide \emph{precise} and \emph{finite-dimensional} results on the inverse sketched matrix.
This addresses the important case of $z=0$, which is typically
avoided in RMT analyses, due to the difficulty of dealing with the
possible singularity. More generally, the resolvent also appears as
the key object of study in the spectrum analysis of linear operators in
general Hilbert space \cite{akhiezer2013theory}, as well as
in modern convex optimization theory
\cite{bauschke2011convex}, thereby showing a much broader
interest of the proposed analysis.

\paragraph{Sketching.}
For overviews of sketching and random projection methods, we refer to
\cite{vempala2005random,tropp2011structure,Mah-mat-rev_JRNL,woodruff2014sketching,DM16_CACM,drineas2017lectures,RandNLA_PCMIchapter_chapter,DM21_NoticesAMS}.   
A key result in this area is the Johnson-Lindenstrauss lemma, which
states that norms, and thus also relative distances between
points, are \emph{approximately} preserved after sketching, i.e.,
$(1-\eta)\|\x_i\|^2\leq\|\S\x_i\|^2\leq(1+\eta)\|\x_i\|^2$ for
$\x_1,\ldots,\x_n\in\mathbb{R}^p$.  
This is further extended to the \emph{subspace embedding property}: for
all $\x$, the norm of $\x$ is preserved up to an $\eta$ factor. 
Subspace embeddings were first used in RandNLA by \cite{drineas2006sampling}, where they were used in a data-aware context to obtain relative-error approximations for $\ell_2$ regression and low-rank matrix approximation \cite{Drineas08CUR}.
Subsequently, data-oblivious subspace embeddings were used by \cite{sarlos2006improved} and popularized by \cite{woodruff2014sketching}.
Both data-aware and data-oblivious subspace embeddings
can be used to derive bounds for the accuracy of various algorithms \cite{DM16_CACM,RandNLA_PCMIchapter_chapter}. 

The most popular sketching methods include random projections with i.i.d.
entries, random sampling of the datapoints, uniform orthogonal
projections, Subsampled Randomized Hadamard Transform (SRHT)
\cite{sarlos2006improved,ailon2006approximate}, leverage score 
sampling~\cite{Drineas08CUR,fast-leverage-scores,MMY15}, and
CountSketch \cite{charikar2002finding,cw-sparse,nn-sparse,mm-sparse}.   
Random projection based approaches have been developed for a wide
variety of problems in data science, statistics, machine learning etc.,
including linear regression
\cite{sarlos2006improved,drineas2011faster,raskutti2014statistical,dobriban2018new},
ridge regression
\cite{lu2013faster,chen2015fast,wang2017sketched,liu2019ridge}, two
sample testing \cite{lopes2011more,srivastava2016raptt},
classification \cite{cannings2017random}, PCA
\cite{frieze2004fast,dkm_matrix2,sarlos2006improved,liberty2007randomized,halko2011algorithm,tropp2011structure,woolfe2008fast,Musco2015,tropp2017practical,7008533,yang2020reduce,gataric2020sparse},
convex optimization
\cite{pilanci2015randomized,pilanci2016iterative,pilanci2017newton},
etc.; see \cite{woodruff2014sketching,DM16_CACM,RandNLA_PCMIchapter_chapter} for a more
comprehensive list. Our new \less\ embeddings have the potential to
be relevant for all those important applications.

\section{Dense Gaussian and sub-gaussian sketches have small inversion bias}
\label{noib}
Consider first the classical Gaussian sketch, i.e., where the entries of $\S$ are
i.i.d.~standard normal scaled by $1/\sqrt m$. In this special case, the
sketched covariance matrix $\tilde\A^\top\tilde\A$ is
a Wishart-distributed random matrix, and we have:
\begin{align}
  \E\big[(\tilde\A^\top\tilde\A)^{-1}\big] = \tfrac
  m{m-d-1}(\A^\top\A)^{-1}\quad\text{for}\quad m\geq d+2.\label{eq:exact-correction}
\end{align}
In other words, even though the sketched inverse covariance is not an
unbiased estimate, the bias can be corrected by simply scaling
the matrix, after which averaging can be used effectively without
encountering any inversion bias.

The key property which enables exact bias-correction for the Gaussian
sketch is \emph{orthogonal invariance}. This property
requires that for any orthonormal matrix $\O$, the distributions of
the random matrices $\S$ and $\S\O$ are identical. An example beyond
Gaussians are Haar sketches, which are uniform over all partial
orthogonal matrices. If a sketch $\S$ is orthogonally invariant and
$\tilde\A^\top\tilde\A$ is invertible with probability one, then we
can show that the inversion bias can be corrected exactly, in that,
\eqref{eq:exact-correction} holds with some constant factor $c$ (replacing 
the factor $\frac m{m-d-1}$) that depends on the distribution of the sketch (see
Proposition \ref{ro} in Appendix \ref{s:orthogonal}).

Exact bias-correction, achieved by the Gaussian
sketch and other orthogonally invariant sketches, is no longer
possible for general sub-gaussian sketches. Here, we consider
sketching matrices with i.i.d.\ entries that (after scaling by $\sqrt m$) have $O(1)$ sub-gaussian
Orlicz norm. Consider for example the so-called Rademacher sketch,
with $\S$ consisting of scaled i.i.d.\ random sign entries (which is useful for reducing the
cost of randomness relative to the Gaussian sketch). In this case, an
exact bias-correction analogous to \eqref{eq:exact-correction} is
clearly infeasible for any $d>1$, simply because, with some positive
(but exponentially small) probability, the matrix $\tilde\A^\top\tilde\A$
will be non-invertible, making the expectation undefined.
Yet, any task where we observe at most polynomially many independent
estimates (such as averaging)
should not be affected by such low-probability events,
so we need a notion of near-unbiasedness that is robust to
this. To that end, we first recall a standard definition of a relative
error approximation for a positive semi-definite matrix.

\begin{definition}[Relative error approximation]
  \label{d:approximation}
  A positive semi-definite (p.s.d.) matrix $\tilde\C$ (or a non-negative scalar) is an
$\epssub$-approximation of $\C$, denoted as
$\tilde\C\approx_\epssub\C$, if
\[\C/(1+\epssub)\preceq\tilde\C\preceq(1+\epssub)\cdot \C.\]
If $\tilde\C$ is random and the above holds with
probability $1-\delta$, then we call it an
$(\epssub,\delta)$-approximation.
\end{definition}
\begin{remark}[Subspace embedding]\label{r:subspace-embedding}
If $\tilde\C=\tilde\A^\top\tilde\A$ where
$\tilde\A\in\R^{m\times d}$ is a sketch of $\A\in\R^{n\times d}$, then
the condition
$\tilde\A^\top\tilde\A\approx_\epssub\A^\top\A$ is called the
subspace embedding property with error $\eta$.
\end{remark}

\noindent 
For instance, any sketching matrix $\S$ with i.i.d.\@ $O(1)$ sub-gaussian random entries, of size
$m=O((d+\ln(1/\delta))/\epssub^2)$, where $\epssub\in(0,1)$, ensures 
that $\tilde\A=\S\A$ with probability $1-\delta$ satisfies the subspace embedding property with
error $\epssub$. In other words, $\tilde\A^\top\tilde\A$ is an
$(\epssub,\delta)$-approximation of $\A^\top\A$ (This is known to
be tight; see, e.g., \cite{nelson2014lower}). As a
consequence, the same guarantee applies to the inverse
$(\tilde\A^\top\tilde\A)^{-1}$, relative to $(\A^\top\A)^{-1}$.
The $\delta$ failure probability makes this definition robust to the rare
events where $\tilde\A^\top\tilde\A$ is not invertible. It is natural
to desire a similar robustness in the definition of
near-unbiasedness. We achieve this as follows.

\begin{definition}[$(\epsilon,\delta)$-unbiased estimator]\label{d:unbiased-estimator}
A random p.s.d.~matrix $\tilde\C$ is an
  $(\epsilon,\delta)$-unbiased estimator of $\C$ if there is an event
  $\mathcal E$ that holds with probability $1-\delta$ such that
  \begin{align*}
    \E\big[\tilde\C\mid \mathcal E\big] \approx_\epsilon
    \C,\quad\text{and}\quad\tilde\C\preceq O(1)\cdot\C\quad\text{when
    conditioned on $\Ec$.}
  \end{align*}
\end{definition}

\noindent
Note that this definition only becomes meaningful if we use it
with an $\epsilon$ that is much smaller than the approximation error $\epssub$ in
Definition \ref{d:approximation} (for instance, we will often have
$\epssub=\Omega(1)$ and $\epsilon\ll 1$).
Further, note the following two important aspects of
Definition~\ref{d:unbiased-estimator}. First, instead of a simple
expectation, we condition on some high probability 
event $\Ec$, which, similarly as in Definition \ref{d:approximation},
allows robustness against such corner cases as when 
the sketch $\tilde\A^\top\tilde\A$ is not invertible. Second,
conditioned on the event $\Ec$, in addition to an
$\epsilon$-approximation holding in expectation, we require a weaker upper 
bound to hold almost surely, in terms of the target matrix $\C$
scaled by some constant factor. 
This condition is important to
guard against certain corner cases where the probability mass is
extremely skewed. For instance, suppose that $\tilde
C$ is a scalar random variable which is uniform over $[0,1]$ and has an
additional probability mass of $10^{-10}$ at the value
$10^{100}$. Here, averaging will not prove effective at converging to
the true expectation of $\tilde C$, but we could still use the notion of
$(\epsilon,\delta)$-unbiasedness to show that the average of
an appropriately chosen number (much smaller than $10^{10}$) of i.i.d. copies will
converge very close to $0.5$, by choosing an event $\Ec$ that avoids
the $10^{100}$ (see Appendix~\ref{s:averaging}).

We are now ready to state our main result for sub-gaussian sketches
(this is in fact a corollary of our more general result,
Theorem~\ref{t:structural}, discussed in Section~\ref{s:structural}), which
asserts that after proper rescaling, not 
only the Gaussian sketch, but in fact all sub-gaussian sketches
(including the Rademacher sketch) enjoy small inversion bias.

\begin{proposition}[Near-unbiasedness of sub-gaussian sketches]\label{t:subgaussian}
Let $\S$ be an $m\times n$ random matrix such that $\sqrt m\,\S$
  has i.i.d.\ $O (1)$-sub-gaussian entries with mean zero and unit
  variance. There is $C=O(1)$ such that for any $\epsilon,\delta\in(0,1)$ if
  $m\geq C(d+\sqrt d/\epsilon + \log(1/\delta))$, then for all
  $\A\in\R^{n\times d}$ of rank $d$, $(\frac
  m{m-d}\A^\top\S^\top\S\A)^{-1}$ is an
  $(\epsilon,\delta)$-unbiased estimator of $(\A^\top\A)^{-1}$.
\end{proposition}

\noindent
Observe that the scaling $\frac m{m-d}$ essentially matches the exact
bias-correction for Gaussian sketches, which is $\frac m{m-d-1}$. In
fact, the same statement of the theorem holds with either scaling, and we
merely chose the simplest form of the scaling. 

As a corollary of the near-unbiasedness of sub-gaussian sketches, we
can show the following approximation guarantee for averaging the
inverse covariance matrix estimates. Recall that our primary
motivation is parallel and
distributed averaging, where the
computational cost does not grow with the number of independent
estimates.

\begin{corollary}\label{c:subgaussian}
Let $\S$ be a sub-gaussian sketching matrix of size $m$, and let
  $\S_1,...,\S_q$ be i.i.d.\ copies of $\S$. There is
  $C=O(1)$ such that if $m\geq C(d + \sqrt d/\epsilon+\log(q/\delta))$
  and $q\geq Cm\log(d/\delta)$, then for any $\A\in\R^{n\times d}$ of
  rank $d$,
  $\frac1q\sum_{i=1}^q(\tfrac{m}{m-d}\A^\top\S_i^\top\S_i\A)^{-1}$ is an
$(\epsilon,\delta)$-approximation of $(\A^\top\A)^{-1}$.
\end{corollary}

\noindent
Proposition \ref{t:subgaussian} shows that for a sub-gaussian sketch
$\tilde\A=\S\A$ of size $m\geq Cd$, 
the sketched inverse covariance $(\frac
m{m-d}\tilde\A^\top\tilde\A)^{-1}$ has inversion bias $O(\sqrt d/m)$. 
This means that the inversion bias of this estimator is smaller than
the approximation error, which is $\Theta(\sqrt{d/m})$, by a factor of  
$O(1/\sqrt m)$. Thus, using Corollary \ref{c:subgaussian}, we can
reduce the approximation error by averaging
$q=O(m\log(d/\delta))$ copies of this estimator,
obtaining that
$\frac1q\sum_{i=1}^q(\tfrac{m}{m-d}\tilde\A_i^\top\tilde\A_i)^{-1}$
is with high probability an $O(\sqrt d/m)$-approximation of
$(\A^\top\A)^{-1}$.  In particular, when
$m=\Theta(d)$, then the approximation error of a single estimate
(without averaging) is
$\Theta(1)$, whereas the approximation error of the averaged estimate
is only $O(1/\sqrt d)$.

\section{Main results: Less inversion bias with \less\ embeddings}

To address the high computational cost of sub-gaussian sketches, while
preserving their good near-unbiasedness properties, 
we propose a new family of sketches, which we call LEverage Score Sparsified (\less)
embeddings. A \less\ embedding is defined simply as a sparsified
sub-gaussian sketch, where the sparsification is designed so as to ensure
small inversion bias for a particular matrix $\A$. Our
approach combines ideas from approximate leverage
score sampling (which is data-aware) with ideas from 
sparse embedding matrices (which are normally data-oblivious).
Importantly, neither strategy by itself is sufficient to ensure
small inversion bias (see our lower bound in Theorem \ref{t:lower} and
discussion in Section \ref{s:structural}). Each row of a \less\ embedding is
sparsified independently using a sparsification pattern defined as
follows. Recall that for a tall full rank matrix $\A$, we use $\a_i^\top$ to denote the $i$th row of
$\A$, and the $i$th leverage score of $\A$ is defined as
$l_i=\a_i^\top(\A^\top\A)^{-1}\a_i$. 

\begin{definition}[\less: LEverage Score Sparsified embedding]\label{d:leverage-score-sparsifier}
Fix a matrix $\A\in\R^{n\times d}$ of rank~$d$ with leverage scores $l_1,...,l_n$, and let 
$s_1,...,s_d$ be sampled i.i.d.~from a
probability distribution $(p_1,...,p_n)$ such that
$p_i\approx_{O(1)} l_i/d$ for all $i$. Then, the
random vector $\xib^\top
=\big(\sqrt{\!\frac{b_1}{dp_1}},...,\sqrt{\!\frac{b_n}{dp_n}}\big)$, where
$b_i=\sum_{t=1}^d1_{[s_t=i]}$, is called a \underline{leverage score sparsifier} for $\A$.

Sketching matrix $\S$ is a \underline{\less\ embedding} of size $m$
for a matrix $\A$, if it consists of $m$ i.i.d.\ row vectors distributed as
$\frac1{\sqrt m}(\x\circ \xib)^\top$, where $\circ$ denotes an entry-wise
product and $\x$ is a random vector with i.i.d.\ mean zero, unit
variance, $O(1)$-sub-gaussian entries. 
\end{definition}

\begin{remark}[Time complexity of LESS]
  Given a matrix $\A\in\R^{n\times d}$ of rank $d$, there is an
  algorithm with an $O(\nnz(\A)\log n + d^3\log d)$ time preprocessing
  step, that can then construct a LESS embedding sketch $\S\A$ of size
  $m$ in time $O(md^2)$. In the following results we always use $m\geq d\log d$, in
  which case the total cost of constructing a LESS embedding is $O(\nnz(\A)\log n+md^2)$.
\end{remark}
\noindent
The matrix product
$\S\A$ costs only $O(md^2)$ because, by definition, the number
of non-zeros per row of $\S$ is bounded almost surely by $d$. 
It is not essential for our analysis that we sample exactly $d$ indices in each row of a 
\less\ embedding, but we fix it here
for the sake of simplicity. 
We could also have approximately $d$ non-zeros per row,
and similar results would still hold. To construct the
distribution $(p_1,...,p_n)$, the sparsifier requires a constant
relative error approximation 
of all the leverage scores of $\A$, which can be computed in
$O(\nnz(\A)\log n + d^3\log d)$ time \cite{fast-leverage-scores,cw-sparse}. Alternatively we can use our
approach in a data-oblivious way, by combining LEverage Score
Sparsification with the Randomized Hadamard Transform \cite{ailon2009fast,drineas2011faster}, which we may abbreviate as \less RHT. Here, the matrix
$\A$ is first transformed so that it has approximately uniform
leverage scores \cite{fast-leverage-scores}, and then we can sparsify it using a uniform
distribution, i.e., $p_i=1/n$ for all $i$, with total cost $O(nd\log n + md^2)$. Finally, computing
the sketched inverse covariance matrix estimator $(\frac 
m{m-d}\A^\top\S^\top\S\A)^{-1}$ only adds an $O(md^2)$ cost. These
costs can be further optimized using fast matrix multiplication
\cite{williams2012multiplying}.%
\footnote{The cost of computing the matrix product $\S\A$ can be
  optimized beyond $O(md^2)$ by adapting the fast matrix
  multiplication routines to take advantage of the sparsity pattern;
  see, e.g., \cite{yuster2005fast}.}

In our main result, we show that \less\ embeddings enjoy small
inversion bias, nearly matching our guarantee for sub-gaussian
sketches (Proposition \ref{t:subgaussian}).

\begin{theorem}[Near-unbiasedness for
  \less]\label{t:lsse}
Suppose that $\S$ is a \less\ embedding of size $m$
for a rank $d$ matrix $\A\in\R^{n\times d}$. 
There is $C=O(1)$ such that
if $m\geq C(d\log(d/\delta) + \sqrt d/\epsilon)$ then the sketch $(\frac
  m{m-d}\A^\top\S^\top\S\A)^{-1}$ is an
  $(\epsilon,\delta)$-unbiased estimator of
  $(\A^\top\A)^{-1}$.
\end{theorem}

\noindent
Thus, we show that the inversion bias guarantee for \less\ embeddings
matches our result for sub-gaussian sketches up to a logarithmic
factor. This additional factor is standard in the analysis of fast
sketching methods. It comes from the fact that, as an artifact of the
matrix concentration bounds \cite{tropp2012user} we use in our analysis of \less\
embeddings, a sketch of size $m=O(d\log d)$ is needed to satisfy
the subspace embedding property, which is one of our two
structural conditions for small inversion bias (see
Section~\ref{s:structural}). As a corollary, we obtain 
an improved guarantee for parallel and distributed averaging of i.i.d.\ sketched inverse
covariance estimates which also matches the corresponding statement
for sub-gaussian sketches (Corollary~\ref{c:subgaussian}) up to
logarithmic factors. 

\begin{corollary}\label{c:less}
Let $\S$ be a \less\ embedding matrix of size $m$ for
a rank $d$ matrix $\A\in\R^{n\times d}$, and let
  $\S_1,...,\S_q$ be i.i.d.\ copies of $\S$. There is
  $C=O(1)$ such that if $m\geq C(d\log(q/\delta) + \sqrt d/\epsilon)$
  and $q\geq Cm\log^2(d/\delta)$, then 
  $\frac1q\sum_{i=1}^q(\tfrac{m}{m-d}\A^\top\S_i^\top\S_i\A)^{-1}$ is an
$(\epsilon,\delta)$-approximation of $(\A^\top\A)^{-1}$.
\end{corollary}

To motivate and place our new algorithm into context, we demonstrate that existing fast
sketching techniques may not achieve an inversion bias bound
comparable to that of sub-gaussian sketches, even if they achieve
a nearly matching subspace embedding guarantee.
This lower bound
demonstrates the hardness of
constructing an $(\epsilon,\delta)$-unbiased estimator of
the inverse covariance matrix from its sketch. 
We show this here for leverage score sampling~\cite{drineas2006sampling,Drineas08CUR,fast-leverage-scores,MMY15}.
However, based on evidence from our analysis, we conjecture that similar lower bounds hold for
other methods such as Subsampled Randomized Hadamard Transform \cite{ailon2009fast,drineas2011faster}
and data-oblivious sparse embedding matrices \cite{cw-sparse,nn-sparse,mm-sparse}.%
\footnote{An alternative approach to achieving small
  inversion bias is to chain together a fast sketch
having a larger size, say, $t=\Ot(d/\epsilon^2)$, with a sub-gaussian
sketch having a smaller size $m=O(d+\sqrt d/\epsilon)$.
However, this leads to a sub-optimal time complexity in terms of
the polynomial dependence on $d$ due to the cost $O(tmd)$ of the sub-gaussian
sketch. For example, with $\epsilon=1/\sqrt d$, the overall cost is $\Ot(\nnz(\A) + d^4)$ compared to $\Ot(\nnz(\A)+d^3)$
with LESS.}

\begin{theorem}[Lower bound for leverage score sampling]\label{t:lower}
For any $n\geq 2d\geq 4$, there is an $n\times d$ matrix $\A$ and
 a row sampling $(p_1,...,p_n)$, with a corresponding $m\times n$
 sketching matrix $\S$, s.t.:
  \begin{enumerate}
    \item The row sampling $(p_1,...,p_n)$ is a $1/2$-approximation of leverage score
      sampling;  and
    \item For any sketch size $m$ and scaling $\gamma$,
      $(\gamma\A^\top\S^\top\S\A)^{-1}$ is \underline{not} an
  $(\epsilon,\delta)$-unbiased estimator of $(\A^\top\A)^{-1}$ with any
  $\epsilon\leq c\frac dm$ and $\delta\leq c(\frac dm)^2$, where $c>0$
  is an absolute constant.
\end{enumerate}
\end{theorem}

\noindent
In the proof of Theorem \ref{t:lower}, we develop a new lower bound for
the inverse moment of the Binomial distribution (Lemma
\ref{l:inverse-moment}), by using anti-concentration of measure via
the Paley-Zygmund inequality, which should be of independent interest.
To illustrate Theorem \ref{t:lower}, consider a sketch of size $m=O(d\log
d)$. This is sufficient to ensure that approximate leverage score
sampling achieves the subspace embedding property with relative error
$O(1)$. In particular, it implies that for any $\gamma=\Theta(1)$,
the inverse covariance matrix estimator $(\gamma\A^\top\S^\top\S\A)^{-1}$ is
with high probability an $O(1)$-approximation of
$(\A^\top\A)^{-1}$. Our lower bound implies that the inversion bias of
any such estimator is $\Omega(1/\log d)$, which is up to logarithmic
factors the same as the approximation achieved by a single
estimator.

Thus, Theorem \ref{t:lower} shows that when $m=O(d\log d)$, averaging i.i.d.\ copies of the
sketched inverse covariance estimator obtained from approximate
leverage score sampling may lead to only $\Omega(1/\log d)$ factor
improvement in the approximation, which is merely inverse-logarithmic in $d$. In
contrast, Theorem \ref{t:lsse} shows that, when using our new \less\
embeddings with the same sketch size and time complexity, averaging
i.i.d.\ copies of the sketched inverse covariance
reduces the approximation error by a factor of  $O(1/\sqrt d)$,
which is inverse-polynomial in $d$ and thus far superior to what is achievable
by approximate leverage score sampling.

\section{Our techniques: Structural conditions for near-unbiasedness}
\label{s:structural}
In order for our analysis of inversion bias to apply to a wide range
of sketching techniques, we give two key structural
conditions for a random sketching matrix $\S$ that are sufficient to achieve provably
small inversion bias. The first is the subspace
embedding property discussed in Remark~\ref{r:subspace-embedding}, 
which we now use as one of the key conditions needed in our analysis.

\begin{condition}[Subspace embedding]\label{cond1}
The (sketching) matrix $\S\in\R^{m\times n}$
satisfies the subspace embedding condition with $\eta\geq 0$
for a matrix $\A\in\R^{n\times d}$, if
$\A^\top\S^\top\S\A\approx_{\eta}\A^\top\A$.
\end{condition}

\noindent
The second structural condition for small inversion bias is a property
of each individual row of~$\S$. We use an $n$-dimensional random
row vector $\x^\top$ to denote the marginal distribution of a row
of $\S$ (after scaling by $\sqrt m$).
This condition represents a key novelty in our analysis.
\begin{condition}[Restricted Bai-Silverstein]\label{cond2}
The random vector $\x\in\R^n$ satisfies the Restricted Bai-Silverstein
condition with $\alpha>0$ for a matrix $\A\in\R^{n\times d}$, if
  $\Var\!\big[\x^\top\B\x\big]\leq \alpha\cdot\tr(\B^2)$
  for all p.s.d.\ matrices $\B$ such that $\B=\P\B\P$, where $\P$ is the
  projection onto the column span of $\A$.  
\end{condition}

\noindent
Based on these two structural conditions, we show the following
result, which we use to prove both Proposition \ref{t:subgaussian} and
Theorem \ref{t:lsse}. In this result, we will refer to an $m\times n$ sketching matrix $\S_m$, indexed by the number of rows $m$.

\begin{theorem}[Structural conditions for near-unbiasedness]\label{t:structural}
Fix $\A\in\R^{n\times d}$ with rank~$d$ and let $\S_m$
consist of $m\geq8d$ i.i.d.~rows distributed as $\frac1{\sqrt
  m}\x^\top$, where $\E[\x\x^\top]=\I_n$. 
Suppose that $\S_{m/3}$ satisfies
Condition \ref{cond1} (subspace embedding) for $\eta=1/2$,  with probability
$1-\delta/3$, where $\delta\leq 1/m^3$. 
Suppose also that
$\x$ satisfies Condition~\ref{cond2} (Restricted Bai-Silverstein) with
some $\alpha\geq 1$.
Then $(\frac m{m-d}\A^\top\S_m^\top\S_m\A)^{-1}$ is an
$(\epsilon,\delta)$-unbiased estimator of $(\A^\top\A)^{-1}$ for
$\epsilon=O(\alpha\sqrt d/m)$.
\end{theorem}

\noindent
The proof of Theorem \ref{t:structural} adapts and extends techniques
for analyzing the limiting Stieltjes transform for high-dimensional
random matrices in the so-called Marchenko-Pastur regime 
(also called the proportional or mean-field limit). 
This
regime arises if we let $n$, $m$ and $d$ all go 
to infinity and let the ratio $m/d$ converge to a fixed constant
larger than unity. Crucially, our analysis is non-asymptotic, and it is
not restricted to the constant aspect ratio between the sketch size
and the dimension. Further, while classical random matrix theory
analysis considers matrix resolvents, which take the form
$(\gamma\A^\top\S_m^\top\S_m\A+z\I)^{-1}$ for $z,\gamma\neq 0$, and are
well-defined with full probability, we consider the case of $z=0$ where
the matrix in question may be undefined with positive probability.
We address this by defining a high probability event which ensures that the sketch $(\frac
m{m-d}\A^\top\S_m^\top\S_m\A)^{-1}$ is well-defined and bounded, while
preserving enough of the independence structure in the conditional
distribution for the expectation analysis to go~through. Specifically,
we split the sketch into three parts, and we condition on the event that each part
satisfies the subspace embedding property. This way, for any pair of
rows, there is a part of the sketch that ensures invertibility while
being independent from the two rows, which is important for the analysis.

\paragraph{Subspace embedding condition.}
Our first structural condition for small inversion bias
(Condition~\ref{cond1})
is a variant of the subspace embedding property, which is standard in the sketching
literature. 
In particular, this condition immediately implies that $(\frac
m{m-d}\A^\top\S_m^\top\S_m\A)^{-1}$ is with probability~$1-\delta$ an
$O(1)$-approximation of $(\A^\top\A)^{-1}$. For sub-gaussian sketches
this is known to hold with sketch size $O(d+\log(1/\delta))$ \cite{nelson2014lower}. We prove this
for \less\ of size $O(d\log(d/\delta))$.
\begin{lemma}[Subspace embedding for \less]\label{l:subspace-embedding}
Suppose that $\S$ is a \less\ embedding of size $m$ for a rank $d$ matrix $\A\in\R^{n\times d}$.
There is $C=O(1)$ such that if  $m\geq
    Cd\log(d/\delta)/\epssub^2$ for $\epssub\in(0,1)$, then the sketch $\A^\top\S^\top\S\A$ is
    an $(\epssub,\delta)$-approximation of $\A^\top\A$.
  \end{lemma}

  \noindent
  The subspace embedding guarantee for \less\ embeddings
is as good as that for existing fast sketching methods.  However, the analysis differs
from the ones used for either data-aware leverage score sampling or for data-oblivious
sparse sketches. We show the result by
deriving a subexponential bound on the matrix moments of a \less\
embedding (Lemma \ref{l:subexponential-moments}), relying on a novel variant of the 
Hanson-Wright concentration inequality for quadratic forms based on
orthogonal projection matrices
(Lemma~\ref{l:restricted-hanson-wright}). We then use this to invoke a
matrix Bernstein inequality for random matrices with subexponential
moments \cite[Theorem~6.2]{tropp2012user}.

\paragraph{Restricted Bai-Silverstein condition.}
Our second structural condition for small inversion bias (Condition~\ref{cond2}) is not 
commonly seen in sketching, but we expect that it will be of broader
interest in adapting high-dimensional random matrix theory to
RandNLA~\cite{precise-expressions,surrogate-design,DM21_NoticesAMS}.
It is based on the classical inequality of Bai and Silverstein \cite{bai2010spectral}
which bounds the deviation of a random quadratic form $\x^\top\B\x$
from its mean. We call it the \emph{Restricted Bai-Silverstein condition}
because, unlike in the classical version, we only require the inequality
to hold for matrices $\B$ that are restricted to the subspace spanned by the
columns of $\A$. By contrast, in classical random matrix theory it is
often assumed that the the following (unrestricted) condition holds.
\begin{condition}[Bai-Silverstein]\label{cond:classical}
Random vector $\x\in\R^n$ satisfies the (unrestricted) Bai-Silverstein
condition with $\alpha>0$, if
  $\Var\!\big[\x^\top\B\x\big]\leq \alpha\cdot\tr(\B^2)$
  for all $n\times n$ p.s.d.\ matrices $\B$.  
\end{condition}

\noindent
 When the random vector $\x$ is $O(1)$-sub-gaussian, then Condition
 \ref{cond:classical} is satisfied with $\alpha=O(1)$, as a
 consequence of the original inequality of \cite{bai2010spectral}.%
\footnote{The
original lemma applies more broadly to higher moments; we cite only
the case relevant to our analysis.}

\begin{lemma}[Bai-Silverstein inequality]\label{l:classical-bai-silverstein}
Let $\x$ have $n$ independent entries with 
mean zero and unit variance such that $\E [x_i^4] =O(1)$. Then,
Condition \ref{cond:classical} is satisfied with $\alpha=O(1)$.
\end{lemma}

\section{Restricted Bai-Silverstein inequality}

The Bai-Silverstein inequality from  Lemma
\ref{l:classical-bai-silverstein} does not directly apply to any of the
fast sketching methods discussed above (see 
Appendix~\ref{s:lower-bai-silverstein} for lower bounds). 
However, we state and prove a generalization of this lemma, which
allows us to show the Restricted Bai-Silverstein condition
(Condition~\ref{cond2}) for our new \less\ embeddings.   

To provide some intuition behind this result, consider the variance
term $\Var[\x^\top\B\x]$ which appears in the Restricted Bai-Silverstein condition,
 where $\frac1{\sqrt m}\x^\top$ represents a random row vector of the sketching matrix
 $\S$. The condition requires that just this one row vector carries
 enough randomness to produce an accurate sketch of the trace of a quadratic
 form $\B$. This is in contrast to the subspace embedding condition,
 which uses the joint randomness of all the rows of $\S$.
Lemma \ref{l:classical-bai-silverstein} achieves this by
enforcing a fourth-moment bound on all 
of the entries of $\x$. Suppose that we sparsify this
vector, following the strategy of sparse embedding matrices, by
multiplying each entry of $\x$ with an independent scaled Bernoulli
variable, obtaining
$\sqrt{\frac ms}\, b_ix_i$ for $b_i\sim\mathrm{Bernoulli}(\frac sm)$, where $s\ll
m$ is the sparsity level and $i$ is the entry index.%
\footnote{Most commonly studied sparse embedding matrices have
  non-independent entries.  However, the i.i.d.\ variant we consider
  offers an equivalent guarantee for the subspace embedding property. See \cite{cohen2016nearly} and
  Appendix \ref{s:lower-bai-silverstein}.}
This
preserves the mean and variance assumptions from
Lemma~\ref{l:classical-bai-silverstein}, but as long as $s=o(m)$, it
violates the fourth-moment assumption. Thus, it is natural to ask
whether we can relax this fourth-moment assumption. It turns out that,
if we do the sparsification in a data-oblivious manner, then the
answer is no, since the random vector 
may not capture most of the relevant directions in the matrix~$\B$
(see Appendix~\ref{s:lower-bai-silverstein}). Importantly, this can
occur even when the rows of the sketch \emph{together} 
capture all of the directions, ensuring the subspace embedding
property, which is already the case when we set the sparsity level to
be as small as $s=O(\log d)$. In other words, there is a wide gap between the
sparsity needed to preserve the Bai-Silverstein inequality,
$s=\Omega(m)$, and sparsity needed to ensure the subspace 
embedding. 

Crucially, Theorem \ref{t:structural} does not require the
Bai-Silverstein inequality to hold for all $n\times n$ p.s.d.\ quadratic forms
$\B$. Rather, it restricts the family of quadratic forms 
to those that lie within the column-span of the $n\times d$ data matrix $\A$. In
particular, this restriction implies that the matrix $\B$ is low-rank (it
has at most rank $d$) and its important directions are captured by the
leverage scores of $\A$. We take advantage of
this additional information to relax the fourth-moment assumptions,
obtaining the following generalization of Lemma
\ref{l:classical-bai-silverstein}, which should be of independent interest. 

\begin{theorem}[Restricted Bai-Silverstein inequality]\label{t:generalized-bai-silverstein}
  Fix a matrix $\A\in\R^{n\times d}$ with rank $d$ and leverage scores $l_i$, and let $\x$
have $n$ independent entries with 
mean zero and unit variance such that $\E x_i^4 \leq C/l_i$.
Then, $\x$ satisfies Condition \ref{cond2} with $\alpha=C+2$ for matrix $\A$.
\end{theorem}

\noindent
By setting $\A=\I_n$, where all leverage scores are 1 and the
restriction on $\B$ is vacuous, we not only recover the statement of Lemma
\ref{l:classical-bai-silverstein}, but also our new analysis uses the
Perron-Frobenius theorem to obtain a tight constant factor in the bound
(see Appendix~\ref{s:bai-silverstein}). However, when $\A$ is a tall
matrix, then the fourth-moment assumption becomes potentially 
much more broadly applicable (for example, when the leverage scores are uniform, we only need
$\E x_i^4\leq  C\cdot n/d$). In particular, consider an 
i.i.d.\ sub-gaussian random vector $\x$ sparsified as follows:
$\x\circ \xib$, where we let $\xi_i=b_i/\sqrt{l_i}$ and 
$b_i\sim\mathrm{Bernoulli}(l_i)$. Then, the entries satisfy the
assumptions of Theorem~\ref{t:generalized-bai-silverstein}, with
expected number of non-zeros equal to $d$. Note that this is different
than the data-oblivious sparsification discussed above, since the
entries of the vector corresponding to large leverage scores are less
likely to be zeroed-out than others. This form of sparsification
is nearly equivalent to the one we use for our \less\ embeddings (see Definition
\ref{d:leverage-score-sparsifier}; our
analysis can be applied to either variant), except that it leads to a
non-deterministic level of sparsification. In
Appendix~\ref{s:less-bai-silverstein} we prove the Restricted Bai-Silverstein 
condition with $\alpha=O(1)$ for a leverage score sparsified vector
constructed as in Definition~\ref{d:leverage-score-sparsifier}, which has
non-independent entries.

\section{Conclusions}
We analyzed the phenomenon of inversion bias in sketching-based
estimation tasks involving the inverse covariance matrix. Inversion
bias is a significant bottleneck in methods that use parallel and
distributed averaging. We showed that certain classical sketching
methods (such as sub-gaussian sketches) have small inversion bias,
while many algorithmically efficient sketches (such as leverage score
sampling) may not provide such a guarantee. Finally, we developed a
new efficient sketching method, called LEverage Score Sparsified
(\less) embeddings, which has small inversion bias and
its computational cost is nearly-linear in the input size.

Estimation of the inverse covariance matrix and its various linear
functionals is motivated by a rich body of literature in statistics,
data science, numerical optimization, machine learning, signal
processing, etc., which we summarized in detail in
Section~\ref{s:overview}. Here, we additionally remark that
the $(\epsilon,\delta)$-approximation guarantee we provide for the averaged estimates of  
the inverse covariance (see Corollaries \ref{c:subgaussian} and \ref{c:less})
immediately implies corresponding approximation guarantees for linear
functionals of the inverse covariance in numerous tasks. In a
distributed environment, one can use this to build a system for
querying such functionals, by aggregating coarse estimates computed locally from
$q$ sketches to produce an improved global estimate with
minimal communication cost. We
illustrate this here for a family of linear functionals of the form
$\tr\,\C(\A^\top\A)^{-1}$, parameterized 
by any p.s.d.\ matrix $\C$, as motivated by applications in statistical
inference (see Section \ref{s:overview}). The claim follows from
Corollary \ref{c:less} by letting $\Q_i=(\frac m{m-d}\A^\top\S_i^\top\S_i\A)^{-1}$.
\begin{corollary}[Querying linear functionals]
For any matrix $\A\in\R^{n\times d}$ and $\epsilon,\delta\in(0,1)$, we can use LESS
embeddings of size $m=O(d\log(d/\epsilon\delta)+\sqrt d/\epsilon)$ to
construct $\Q_1,...,\Q_q\in\R^{d\times d}$ in parallel 
time $O(\nnz(\A)\log n + md^2)$, where $q=
   O(m\log^2(d/\delta))$, so that with probability $1-\delta$:
\begin{align*}
  \text{For all p.s.d. matrices $\C\in\R^{d\times d}$,}\qquad
  \frac1q\sum_{i=1}^q\tr\,\C\Q_i\, 
  \approx_\epsilon\, \tr\,\C(\A^\top\A)^{-1}.
  \end{align*}
\end{corollary}


In the context of distributed optimization, our results can be directly applied to
show improved convergence guarantees, for instance, in the case of the
Distributed Iterative Hessian Sketch algorithm 
\cite{pilanci2016iterative,debiasing-second-order} and Distributed
Newton Sketch method
\cite{distributed-newton,determinantal-averaging}. Here, the quantity
of interest is of the form $(\A^\top\A)^{-1}\b$ for some vector $\b$
(where $\A^\top\A$ corresponds to the Hessian and $\b$ corresponds to
the gradient). For those methods, an $\epsilon$-approximation guarantee 
for the average of the sketched inverse covariance matrices, as in
Corollaries \ref{c:subgaussian} and \ref{c:less}, directly implies
that the iterates $\x_t$ produced by the algorithms achieve a
convergence rate of the form $\Delta_{t}\leq
O(\epsilon^t)\cdot\Delta_0$, where $\Delta_t$ 
represents distance from the optimum in the $t$-th iteration. We
illustrate this by applying Corollary \ref{c:less} 
to the existing analysis of Distributed Newton Sketch, as
outlined in Section 4 of \cite{debiasing-second-order}, obtaining an
improved linear-quadratic convergence rate for distributed empirical risk
minimization.
 \begin{corollary}[Distributed Newton Sketch]
   Consider a twice differentiable convex function of the form $f(\x)=\frac1n \sum_{i=1}^n\ell_i(\x^\top\phi_i)
   +\frac\lambda2\|\x\|^2$, where $\x\in\R^d$ and $\phi_i^\top$ is the
   $i$th row of an $n\times d$ data matrix $\Phi$. Given $\x_t$, we
   can use LESS embeddings to construct $q$ independent randomized estimates
   $\Hbh_1(\x_t),...,\Hbh_q(\x_t)$ of the Hessian $\nabla^2f(\x_t)$ in
   parallel time $O(\nnz(\Phi)\log n +md^2)$, where
   $m=O(d\log(d/\epsilon\delta)+\sqrt d/\epsilon)$ and $q=
   O(m\log^2(d/\delta))$, so that 
   \vspace{-2mm}
   \begin{align*}
\x_{t+1} &= \x_t - \frac1q\sum_{i=1}^q \Hbh_i(\x_t)^{-1}\nabla f(\x_t)\quad\text{with probability $1-\delta$ satisfies:}\\
     \|\x_{t+1}-\x^*\|&\leq \max\big\{\epsilon\cdot\sqrt\kappa\|\x_{t+1}-\x^*\|,\ \tfrac{2L}{\lambda_{\min}}\|\x_{t+1}-\x^*\|^2\big\}\quad\text{for}\quad\x^*=\argmin_\x f(\x),
   \end{align*}
\vspace{-4mm}

\noindent
 where $\kappa$, $L$, $\lambda_{\min}$ are the condition number, Lipschitz constant and smallest eigenvalue of $\nabla^2f(\x)$.
 \end{corollary}
This result provides an improvement over the recently proposed DPP-based sketching methods of
\cite{debiasing-second-order}, which suffer no inversion bias but are
more expensive, as well as over other fast sketching methods like row
sampling \cite{distributed-newton}, which, as shown 
in this work, may indeed suffer from large inversion bias.

\subsection*{Acknowledgments}
We would like to acknowledge DARPA, IARPA, NSF, and ONR via its BRC on
RandNLA for providing partial support of this work.  Our conclusions
do not necessarily reflect the position or the policy of our sponsors,
and no official endorsement should be inferred. 

\bibliography{../liao,../asy/references,../pap}
\ifisarxiv
\bibliographystyle{alpha}
\fi

\appendix

\section{Preliminaries}

\paragraph{Notations.} In the remainder of the article, we follow the
convention of denoting scalars by lowercase, vectors by lowercase
boldface, and matrices by uppercase boldface letters. The norm $\|
\cdot \| $ is the Euclidean norm for vectors and the spectral or
operator norm for matrices, and $\| \cdot \|_F$ is the Frobenius norm
for matrices. For vector $\v \in \R^d$, we let $\| \v \|_1 : = \sum_{i=1}^d
|v_i|$ denote the $\ell_1$ norm and $\| \v\|_\infty := \max_i |v_i|$ denote the
$\ell_\infty$ norm of $\v$. We use $\lambda_{\max}(\A)$ to denote the
maximum eigenvalue of a symmetric matrix $\A$. We say $\A \preceq \B$
if and only if $\B- \A$ is positive semi-definite. We use $\A \circ
\B$ to denote the entry-wise Hadamard product of matrices or
vectors. For random vectors or matrices, we say $\A \stackrel{d}{=} \B$ if $\A$
follows the same distribution as $\B$. For positive
semi-definite (p.s.d.) matrices $\A$ and $\B$, or non-negative scalars
$a$ and $b$, we use $\A\approx_{\eta}\B$ and $a\approx_{\eta}b$ to
denote the relative error approximation (Definition
\ref{d:approximation}). The big-O notation is used to absorb constant
factors in upper bounds, where the constant only depends on other big-O
constants appearing in a given statement (thus, all constants
can be made absolute). 

\medskip

An important linear algebraic result that will be used in proving the
restricted Bai-Silverstein inequality
(Theorem~\ref{t:generalized-bai-silverstein}) is the following 
Perron-Frobenius theorem on non-negative matrices. While the most well known version of the Perron-Frobenius theorem concerns matrices with strictly positive entries, there is also a version for matrices with only non-negative entries.
\begin{lemma}[Perron-Frobenius theorem, {\cite[claims~8.3.1~and~8.3.2]{meyer2000matrix}}]\label{l:Perron-Frobenius}
For a non-negative symmetric matrix $\A \in \R^{n \times n}$ such that $[\A]_{ij} \geq 0$ for all $i,j \in \{1, \ldots, n\}$, then the largest eigenvalue of $\A$ is non-negative, i.e., $r=\lambda_{\max}(\A) \ge 0$. Moreover, there is a corresponding eigenvector $\z$, i.e., $\A\z =r\z$  with non-negative entries $\z_i \ge 0$ for all $i$.
\end{lemma}

Our analysis of the inversion bias (proof of Theorem
\ref{t:structural}) crucially relies on a standard
rank-one update formula for the matrix inverse, which is given below.
\begin{lemma}[Sherman-Morrison formula]\label{lem:rank-one}
For an invertible matrix $\A \in \R^{n \times n}$ and $\u,\v \in \R^n$, $\A + \u \v^\top$ is invertible if and only if $1+\v^\top \A^{-1} \u \neq 0$. If this holds, then
\begin{equation*}
    (\A + \u \v^\top)^{-1} = \A^{-1} - \frac{\A^{-1} \u \v^\top \A^{-1} }{1+\v^\top \A^{-1} \u}.
\end{equation*}
In particular, it follows that: 
\begin{equation*}
    (\A + \u \v^\top)^{-1} \u = \frac{\A^{-1} \u}{1+\v^\top \A^{-1} \u}.
\end{equation*}
\end{lemma}

\medskip

Our proofs rely on different types of concentration and
anti-concentration inequalities, from scalars to quadratic forms of the
type $\x^\top \B \x$, and eventually to matrix concentration
bounds. These technical lemmas are collected in this section and will
be repeatedly used in the proofs of our main results.  

\subsection{Scalar concentration and anti-concentration inequalities}

The Burkholder inequality \cite{burkholder1973distribution} provides
moment bounds on the sum of a martingale difference sequence. It is
used to show Lemma~\ref{l:applying-burkholder} as part of the proof of
Theorem \ref{t:structural}.
\begin{lemma}[Burkholder inequality, \cite{burkholder1973distribution}]
For $\{ x_j \}_{j=1}^m$ a real martingale difference sequence with
respect to the increasing $\sigma$ field $\mathcal F_j$, we have, for
$L > 1$, there exists $C_L > 0$ such that 
\[
  \E \bigg[
    \Big| \sum_{j=1}^m x_j \Big|^L
  \bigg]
  \leq C_L \cdot \E \bigg[
    \Big( \sum_{j=1}^m |x_j|^2 \Big)^{L/2}
  \bigg].
\]
\end{lemma}

The Paley-Zygmund inequality is used to establish an
anti-concentration inequality for the Binomial distribution
(Lemma~\ref{l:inverse-moment}), which is the key in
deriving a lower bound for the inversion bias of  leverage score
sampling in Appendix~\ref{s:lower}. 
\begin{lemma}[Paley-Zygmund inequality, {\cite{paley1932note}}]
  For any non-negative variable $Z $ with finite variance and
  $\theta\in(0,1)$, we have:
  \begin{align*}
    \Pr\big(Z\geq \theta\,\E[Z]\big) \geq (1-\theta)^2\frac{\E[Z]^2}{\E[Z^2]}.
  \end{align*}
\end{lemma}

\subsection{Quadratic form concentration}

Being the key object of (one of) the structural conditions in
Theorem~\ref{t:structural}, the (random) quadratic form of the type
$\x^\top \B \x$ will consistently appear in our analysis, for instance
in the form of the Bai-Silverstein inequality in
Lemma~\ref{l:classical-bai-silverstein} on quadratic form variance, as
well as the following Hanson-Wright inequality on the tail
probability.
\begin{lemma}[Hanson-Wright inequality, {\cite[Theorem~1.1]{rudelson2013hanson}}]
Let $\x$ have independent $O (1)$-sub-gaussian entries with
mean zero and unit variance. Then, there is $c=\Omega(1)$ such that
for any $n\times n$ matrix $\B$ and $t\geq 0$,
\begin{align*}
  \Pr\Big\{|\x^\top\B\x-\tr(\B)|\geq t\Big\}\leq
  2\exp\bigg(-c\min\Big\{\frac{t^2}{\|\B\|_F^2},\frac{t}{\|\B\|}\Big\}\bigg).
\end{align*}
\end{lemma}

\subsection{Matrix concentration inequalities}

When random matrices are considered, different variants of Matrix Chernoff/Bernstein inequalities will be needed to handle the case where the random matrix under study is known to have (almost surely) bounded operator norm, or only to admit a subexponential decay for its higher order moments.

\begin{lemma}[Matrix Bernstein: Bounded Case, {\cite[Theorem~1.4]{tropp2012user}}]
For $i=1,2,...$, consider a finite sequence $\X_i$ of $d\times d$
independent and symmetric random matrices such that   
\[
  \E[\X_i] = \mathbf{0}, \quad \lambda_{\max}(\X_i)\leq
        R\quad\text{almost surely.}
\]
Then, defining the variance parameter $\sigma^2 = \| \sum_i\E[\X_i^2] \|$, for any $t>0$ we have:
\begin{align*}
  \Pr \bigg\{ \lambda_{\max}\Big( \sum\nolimits_i \X_i \Big) \geq t
  \bigg\}& \leq d \cdot \exp \left( \frac{ -t^2/2 }{ \sigma^2 +
            R t/3 } \right).
\end{align*}
\end{lemma}

\begin{lemma}[Matrix Bernstein: Subexponential Case, {\cite[Theorem~6.2]{tropp2012user}}] 
For $i=1,2,...$, consider a finite sequence $\X_i$ of $d\times d$
independent and symmetric random matrices such that   
\[
  \E[\X_i] = \mathbf{0}, \quad \E[\X_i^p] \preceq \frac{p!}2
        \cdot R^{p-2} \A_i^2\quad\textmd{for}\quad p=2,3,...
\]
Then, defining the variance parameter $\sigma^2 = \| \sum_i
\A_i^2 \|$, for any $t>0$ we have:
\begin{align*}
  \Pr \bigg\{ \lambda_{\max}\Big( \sum\nolimits_i \X_i \Big) \geq t
  \bigg\}& \leq d \cdot \exp \left( \frac{ -t^2/2 }{ \sigma^2 +
            R t } \right).
\end{align*}
\end{lemma}

\begin{lemma}[Matrix Chernoff, {\cite[Theorem~1.1 and Remark 5.3]{tropp2012user}}]
    For $i=1,2,...$, consider a finite sequence $\X_i$ of $d\times d$
independent positive semi-definite random matrices such that
$\E\big[\sum_i\X_i\big] = \I$ and $\|\X_i\|\leq R$. Then, for any $t\geq \ee$,
we have:
\begin{align*}
  \Pr\Big\{\Big\|\sum_i\X_i\Big\|\geq t\Big\}\leq
  d\cdot\Big(\frac\ee t\Big)^{t/R}.
\end{align*}
\end{lemma}

\section{Structural conditions for small inversion bias}

In this section, we prove Theorem \ref{t:structural}, which gives two
structural conditions for a random sketch of a rank $d$ matrix
$\A\in\R^{n\times d}$ to have small inversion bias. We assume that
the sketching matrix $\S_m\in\R^{m\times n}$ consists of $m\geq
8d$ i.i.d.~rows $\frac1{\sqrt m}\x_i^\top$, where
$\E[\x_i\x_i^\top]=\I$. To simplify the analysis, we assume that $m$ is divisible by $3$.

\subsection{Proof of Theorem \ref{t:structural}}
Note that the subspace embedding assumption (based on Condition
\ref{cond1}) immediately implies the 
result with $\epsilon=O(1)$, so without loss of generality we can
assume that $\alpha\sqrt d/m\leq1$.
  Let $\H=\A^\top\A$ and $\Q=(\gamma\A^\top\S_m^\top\S_m\A)^{-1}$ for
  $\gamma = \frac m{m-d}$. Moreover, let  $\S_{-i}$ denote $\S_m$
  without the $i$th row, with
  $\Q_{-i}=(\gamma\A^\top\S_{-i}^\top\S_{-i}\A)^{-1}$. Finally, for
  $t=m/3$, we define the following events:
  \begin{align}
    \Ec_j: \
    \frac1t\A^\top\Big(\sum_{i=t(j-1)+1}^{tj}\x_{i}\x_{i}^\top\Big)\A\succeq
    \frac12\cdot\A^\top\A,\quad j=1,2,3,\qquad \Ec=\bigwedge_{j=1}^3 \Ec_j.\label{eq:event}
  \end{align}
For each $j$, the meaning of the event $\Ec_j$ is that the average of the rank one matrices $\x_{i}\x_{i}^\top$ over the corresponding $j$-th third of indices $1,\ldots,m$ forms a sketch for $\A$ that is a ``lower'' spectral approximation of $\A^\top \A$.

Note that events $\Ec_1$, $\Ec_2$ and $\Ec_3$ are independent, and for
each $i\in\{1,...,m\}$ there is a $j=j(i)\in\{1,2,3\}$ such 
that:
\begin{enumerate}
\item $\Ec_j$ is independent of $\x_i$; and
\item $\Ec_j$ implies that
  $\Q_{-i}\preceq\gamma\Q_{-i}=(\A^\top\S_{-i}^\top\S_{-i}\A)^{-1}\preceq
  6\cdot(\A^\top\A)^{-1}=6\cdot \H^{-1}$.
\end{enumerate}
Here we use that $\A^\top\S_m^\top\S_m\A$ is the average of the three matrices to which the conditions in $\Ec_j$ refer to, and also that $m\geq 2d$.

From the subspace embedding assumption and the union bound
  we conclude that $\Pr(\Ec)\geq 1-\delta$.  
Letting $\E_{\Ec}$ denote the expectation conditioned on $\Ec$ and
  $\gamma_i=1+\frac{\gamma}{m}\x_i^\top\A\Q_{-i}\A^\top\x_i$, we have: 
  \begin{align*}
   \I - \E_{\Ec}[\Q]\H
    &= -\E_{\Ec}[\Q]\H + \gamma\,\E_{\Ec}[\Q\A^\top\S_m^\top\S_m\A]
=-\E_{\Ec}[\Q]\H + \gamma\,\E_{\Ec}[\Q\A^\top\x_i\x_i^\top\A]
    \\
    &\overset{(*)}=
-\E_{\Ec}[\Q]\H +
      \gamma\,\E_{\Ec}\Big[\tfrac{\Q_{-i}\A^\top\x_i\x_i^\top\A}
      {1+\frac{\gamma}{m}\x_i^\top\A\Q_{-i}\A^\top\x_i} \Big]
      \\
    &=-\E_{\Ec}[\Q]\H
+\E_{\Ec}[\Q_{-i}\A^\top\x_i\x_i^\top\A]
+ \E_{\Ec}\big[(\tfrac{\gamma}{\gamma_i} -1)\Q_{-i}\A^\top\x_i\x_i^\top\A\big]
    \\
    &=
\underbrace{\E_{\Ec}[\Q_{-i}\A^\top(\x_i\x_i^\top-\I)\A]}_{\Z_0}+
      \underbrace{\E_{\Ec}[\Q_{-i}-\Q]\H}_{\Z_1} +
      \underbrace{\E_{\Ec}\big[(\tfrac{\gamma}{\gamma_i}
      -1)\Q_{-i}\A^\top\x_i\x_i^\top\A\big]}_{\Z_2},
  \end{align*}
for a fixed $i$, where $(*)$
  uses the Sherman-Morrison rank-one update formula
  (Lemma~\ref{lem:rank-one}). We also
  used the fact that due to symmetry in the definition of event $\Ec$,
  the marginal distributions of the random vectors $\x_i$ are identical
   after conditioning
  (even though they are no longer independent and identically distributed). 
   To obtain the result, it suffices to bound:
  \begin{align}
    \|\I-\H^{\frac12}\E_{\Ec}[\Q]\H^{\frac12}\|
    &=\|\H^{\frac12}(\Z_0+\Z_1+\Z_2)\H^{-\frac12}\|\nonumber
    \\
    &\leq
      \|\H^{\frac12}\Z_0\H^{-\frac12}\|
      +\|\H^{\frac12}\Z_1\H^{-\frac12}\|+\|\H^{\frac12}\Z_2\H^{-\frac12}\|.
    \label{eq:two-terms}
  \end{align}
  We start by bounding the first term. Without loss of generality,
  assume that events $\Ec_1$ and $\Ec_2$ are both independent of
  $\x_i$, and let $\Ec'=\Ec_1\wedge\Ec_2$ as well as $\delta_3=\Pr(\neg\Ec_3)$. We have:
  \begin{align*}
    \Z_0
    &=
    \frac1{1-\delta_3}\cdot\Big(\E_{\Ec'}[\Q_{-i}\A^\top(\x_i\x_i^\top-\I)\A]
    -\E_{\Ec'}[\Q_{-i}\A^\top(\x_i\x_i^\top-\I)\A\cdot\one_{\neg \Ec_3}]\Big)
    \\
    &=-\frac{1}{1-\delta_3}\cdot
\E_{\Ec'}\big[\Q_{-i}\A^\top(\x_i\x_i^\top-\I)\A\cdot\one_{\neg \Ec_3}\big].
  \end{align*}
Above, we evaluated the expectation $\E_{\Ec'}[\Q_{-i}\A^\top(\x_i\x_i^\top-\I)\A]$ by first conditioning on all randomness except $\x_i$, and using the independence of $\x_i$ and $\Ec'$, as well as $\E[\x\x^\top]=\I$.

  Thus, since $\delta_3\leq \frac12$,  we obtain that:
  \begin{align*}
    \|\H^{\frac12}\Z_0\H^{-\frac12}\|
    &\leq 2\,\Big\|\E_{\Ec'}\big[\H^{\frac12}\Q_{-i}\A^\top(\x_i\x_i^\top-\I)\A\H^{-\frac12}
      \cdot\one_{\neg \Ec_3}\big]\Big\|
      \\
&\leq
2\,\E_{\Ec'}\Big[\big\|\H^{\frac12}\Q_{-i}\A^\top(\x_i\x_i^\top-\I) \A\H^{-\frac12}\big\|
           \cdot\one_{\neg\Ec_3}\Big]
    \\
    &\leq
2\,\E_{\Ec'}\Big[\|\H^{\frac12}\Q_{-i}\H^{\frac12}\|\cdot\big\|\H^{-\frac12}\A^\top(\x_i\x_i^\top-\I) \A\H^{-\frac12}\big\|
      \cdot\one_{\neg\Ec_3}\Big]
\\
&\leq
12\,\E_{\Ec'}\Big[\big(\x_i^\top\A\H^{-1}\A^\top\x_i+1\big)
           \cdot\one_{\neg\Ec_3}\Big].
  \end{align*}
Note that $\E[\x_i^\top\A\H^{-1}\A^\top\x_i]=d$, and using
Condition \ref{cond2} (Restricted Bai-Silverstein), we have
$\Var[\x_i^\top\A\H^{-1}\A^\top\x_i]$  $\leq \alpha\cdot d$ (and both are
still true after conditioning on $\Ec'$, because it is independent of $\x_i$). Chebyshev's inequality thus
implies that for $x\geq 2d$ we have
$\Pr(\x_i^\top\A\H^{-1}\A^\top\x_i\geq x \mid\Ec')\leq C\alpha d/x^2$.
Combining this with the assumption that $\delta_3 \leq \delta \leq1/m^3$, we have:
\begin{align*}
  \E_{\Ec'}\big[\x_i^\top\A\H^{-1}\A^\top\x_i\cdot\one_{\neg\Ec}\big]
  &=\int_0^\infty\Pr(\x_i^\top\A\H^{-1}\A^\top\x_i\cdot\one_{\neg\Ec}\geq
    x\mid\Ec') \,dx
    \\
  &\leq 2m^2\delta_3 +
    \int_{2m^2}^\infty\Pr(\x_i^\top\A\H^{-1}\A^\top\x_i\geq x) \,dx
  \\
  &  \leq \frac 2m +  C\alpha d   \int_{2m^2}^\infty\frac1{x^2}\,dx
    \leq \frac 2m + C\frac {\alpha d}{m^2},
\end{align*}
which implies that $\|\H^{\frac12}\Z_0\H^{-\frac12}\| =
O(1/m + \alpha d/m^2)=O(\alpha\sqrt d/m)$. 
We now move on to bounding the second term in \eqref{eq:two-terms}.
In the following, we will use the observation that for a p.s.d.~random
matrix $\C$ (or non-negative random variable) in the probability space of $\S_m$, we have:
\begin{align}
  \E_{\Ec}[\C] = \frac{\E[(\prod_{j=1}^3\one_{\Ec_j})\cdot \C]}{\Pr(\Ec)}
  \preceq \frac1{1-\delta}\,\E[\one_{\Ec'}\cdot\C]\preceq
  2\cdot\E_{\Ec'}[\C].
  \label{eq:conditional-bound}
\end{align}
Using the above, and the fact that event $\Ec'$ is independent of
$\x_i$, we have:
  \begin{align*}
    \E_{\Ec}[\Q_{-i}-\Q]
   &\preceq 2\cdot \E_{\Ec'}[\Q_{-i}-\Q]
  =  \frac{2\gamma}{m}\cdot \E_{\Ec'}\big[\gamma_i^{-1}\Q_{-i}\A^\top\x_i\x_i^\top\A\Q_{-i}\big]
\preceq\frac{2\gamma}m\cdot \E_{\Ec'}[\Q_{-i}\H\Q_{-i}].
  \end{align*}
 We now bound the second term in \eqref{eq:two-terms} by using the
 fact that $\Ec'$ implies
 $\H^{\frac12}\gamma\Q_{-i}\H^{\frac12}\preceq 6\I$:
  \begin{align*}
    \|\H^{\frac12}\Z_1\H^{-\frac12}\| =
    \|\H^{\frac12}\E_{\Ec}[\Q_{-i}-\Q]\H^{\frac12}\| \leq
    \frac{2\gamma}m\cdot
    \E_{\Ec'}\big[\|\H^{\frac12}\Q_{-i}\H^{\frac12}\cdot\H^{\frac12}\Q_{-i}\H^{\frac12}\|\big]
    \leq\frac {2}m\cdot36.
  \end{align*}
We next bound the last term in \eqref{eq:two-terms}, applying
the Cauchy-Schwarz inequality twice:
\begin{align*}
  \| \H^{\frac12} \Z_2 \H^{-\frac12}\|
  &= \sup_{\| \v \| = 1,~\| \u \| = 1} \E_{\Ec}
    \Big[ \left|\tfrac\gamma{\gamma_i}-1 \right|\cdot \v^\top \H^{\frac12}
    \Q_{-i} \A^\top \x_i \x_i^\top \A \H^{-\frac12} \u \Big] \\
  & \leq \sqrt{\E_{\Ec} \big[ (\tfrac\gamma{\gamma_i} -1)^2\big] } \cdot
    \sup_{\| \v \|=1,~\| \u \| = 1}
    \sqrt{\E_{\Ec}\big[ (\v^\top \H^{\frac12} \Q_{-i} \A^\top \x_i \cdot
    \x_i^\top \A \H^{-\frac12} \u)^2\big] }
  \\
  &\leq \underbrace{\sqrt{\E_{\Ec}\big[(\gamma_i-\gamma)^2\big]}}_{T_1}
    \cdot
    \underbrace{\sup_{\|\u\|=1}\sqrt[4]{\E_{\Ec}\big[(\u^\top\H^{\frac12}\Q_{-i}\A^\top\x_i)^4\big]}}_{T_2}
    \cdot
    \underbrace{\sup_{\|\u\|=1}\sqrt[4]{\E_{\Ec}\big[(\u^\top\H^{-\frac12}\A^\top\x_i)^4\big]}}_{T_3}.
\end{align*}
To bound $T_3$, we rely on Restricted Bai-Silverstein with
$\B=\A\H^{-\frac12}\u\u^\top\H^{-\frac12}\A^\top$, noting that
$\tr(\B^2)=\tr(\B)=(\u^\top\H^{-\frac12}\H\H^{-\frac12}\u)^2=\|\u\|^4=1$.
Recall
that event $\Ec'$ is independent of $\x_i$, so we have:
\begin{align*}
  \E_{\Ec}\big[(\u^\top\H^{-\frac12}\A^\top\x_i)^4\big]
  &\leq 2\,\E_{\Ec'}\big[(\u^\top\H^{-\frac12}\A^\top\x_i)^4\big]
    \\
  &=2\,\E\big[(\x_i^\top\B\x_i)^2\big]
  \\
  &=
    2\Var[\x_i^\top\B\x_i]+2\big(\E[\x_i^\top\B\x_i]\big)^2
  \\
  &\leq 2\alpha\cdot \tr(\B^2)+2\big(\tr(\B)\big)^2 = 2(\alpha+1),
\end{align*}
obtaining that $T_3=O(\sqrt[4]{\alpha+1})$. We can similarly bound $T_2$ by letting
$\B=\A\Q_{-i}\H^{\frac12}\u\u^\top\H^{\frac12}\Q_{-i}\A^\top$. Note
that, conditioned on $\Ec'$, we again have
\[\tr(\B^2)=
\big(\u^\top(\H^{\frac12}\Q_{-i}\H^{\frac12})^2\u\big)^2\leq 6^4,\]
so analogously as above we conclude
that $T_2=O(\sqrt[4]{\alpha+1})$.

It thus remains to bound the term $T_1$. First, note that:
\begin{align}
  \E_{\Ec}[(\gamma - \gamma_i)^2]
  \leq 2\,\E_{\Ec'}[(\gamma - \gamma_i)^2]= 2\,(\gamma - \bar\gamma)^2
  + 2\,\E_{\Ec'}[(\gamma_i-\bar\gamma)^2],\label{eq:gamma}
\end{align}
where $\bar\gamma=\E_{\Ec'}[\gamma_i] = 1+\frac\gamma m\tr(\E_{\Ec'}[\Q_{-i}]\H)$.
To bound the second term in \eqref{eq:gamma}, we write:
\begin{align*}
  \E_{\Ec'}[(\gamma_i-\bar\gamma)^2]
  &= \frac{\gamma^2}{m^2}\,\E_{\Ec'}\Big[\big(\tr (\Q_{-i} -
    \E_{\Ec'}[\Q_{-i}])\H \big)^2\Big]
    + \frac{\gamma^2}{m^2}\,\E_{\Ec'}\Big[ \big(\tr(\Q_{-i}\H) - \x_i^\top \A
    \Q_{-i} \A^\top \x_i \big)^2\Big].
\end{align*}
The latter term can be bounded by again using Condition \ref{cond2},
with $\B=\A\Q_{-i}\A^\top$, obtaining:  
\begin{align*}
\frac{\gamma^2}{m^2}\,\E_{\Ec'}\Big[ \big(\tr(\Q_{-i}\H) - \x_i^\top \A
  \Q_{-i} \A^\top \x_i \big)^2\Big]
  &\leq
    \frac{1}{m^2}\,\E_{\Ec'}\big[\alpha\cdot\tr((\gamma \Q_{-i}\H)^2)\big]
\leq 36\cdot \frac{\alpha d}{m^2}.
\end{align*}
The former term can be bounded using the Burkholder inequality for
martingale difference sequences. We state this bound as a lemma, proven
separately in Appendix \ref{s:applying-burkholder}.
\begin{lemma}\label{l:applying-burkholder}
Let $\Var_{\Ec'}[\cdot]$ be the conditional variance
with respect to event $\Ec'=\Ec_1\wedge\Ec_2$, see \eqref{eq:event},
with $\x_i$ independent of $\Ec'$. Then, there is an
absolute constant $C>0$ such that: 
  \begin{align*}
    \Var_{\Ec'}\!\big[\tr(\Q_{-i}\H)\big] \leq C\cdot d.
  \end{align*}
\end{lemma}

\noindent
Using Lemma \ref{l:applying-burkholder}, we conclude that
$\E_{\Ec'}[(\gamma_i-\bar\gamma)^2]\leq C'\cdot \alpha d/m^2$ for some
absolute constant $C'$. It remains to bound the term:
\begin{align*}
  |\gamma-\bar\gamma| = \bigg|\frac m{m-d} -
  \Big(1+\frac{\gamma}{m}\tr(\E_{\Ec'}[\Q_{-i}]\H)\Big)\bigg| =
\frac{|d-\tr(\E_{\Ec'}[\Q_{-i}]\H)|}{m-d}. 
\end{align*}
Observe that we have:
\begin{align*}
  \big|d-\tr\,\E_{\Ec'}[\Q_{-i}]\H\big|
&=    \big|\tr((\E_{\Ec}[\Q]-\E_{\Ec'}[\Q_{-i}])\H)+ \tr(\I-\E_{\Ec}[\Q]\H)\,\big|
\\
  &=\big|\tr((\E_{\Ec}-\E_{\Ec'})[\Q_{-i}]\H) +\tr(-\Z_1) +
    \tr(\Z_0+\Z_1+\Z_2)\big|
  \\
  &\leq \big|\tr((\E_{\Ec}-\E_{\Ec'})[\Q_{-i}]\H)\big| + |\tr(\Z_0)| + |\tr(\Z_2)|.
\end{align*}
The first two terms can be bounded similarly as we did
$\|\H^{\frac12}\Z_0\H^{-\frac12}\|$, obtaining that
$|\tr(\Z_0)|=O(\alpha d/m)$, and also:
\begin{align*}
  \big|\tr((\E_{\Ec}-\E_{\Ec'})[\Q_{-i}]\H)\big|
  = \frac{\delta_3}{1-\delta_3}\big|\tr((\E_{\Ec'}[\Q_{-i}]-\E_{\Ec'}[\Q_{-i}\mid\neg\Ec_3])\H)\big|=O(d\delta_3)=O(d/m^3).
\end{align*}
For the last term, we have:
\begin{align*}
  |\tr(\Z_2)|
  &=\Big|\E_{\Ec}\big[(\tfrac\gamma{\gamma_i}-1)\x_i^\top\A\Q_{-i}\A^\top\x_i\big]\Big|
  \\
  &\leq\Big|\E_{\Ec}\big[\tfrac{\gamma-\bar\gamma}{\gamma_i}\,
    \x_i^\top\A\Q_{-i}\A^\top\x_i\big]\Big| +
    \Big|\E_{\Ec}\big[\tfrac{\bar\gamma-\gamma_i}{\gamma_i}\,
    \x_i^\top\A\Q_{-i}\A^\top\x_i\big]\Big|
  \\
  &\leq
    |\gamma-\bar\gamma|\cdot\E_{\Ec}[\x_i^\top\A\Q_{-i}\A^\top\x_i]
    + (m-d)\cdot\E\big[|\gamma_i-\bar\gamma|\big]
  \\
  &\leq |\gamma-\bar\gamma|\cdot \frac6{1-\delta}\,
    \E_{\Ec'}[\x_i^\top\A\H^{-1}\A^\top\x_i]+
    (m-d)\cdot\sqrt{\E[(\gamma_i-\bar\gamma)^2]}
  \\
  &\leq |\gamma-\bar\gamma|\cdot \frac 6{1-\delta}\, d+
    \sqrt{C'\alpha d}.
\end{align*}
The bound for the second term $\Big|\E_{\Ec}\big[\tfrac{\bar\gamma-\gamma_i}{\gamma_i} \x_i^\top\A\Q_{-i}\A^\top\x_i\big]\Big|$ 
comes from the definition of $\gamma_i=1+\frac{\gamma}{m}\x_i^\top\A\Q_{-i}\A^\top\x_i$,  because $\x_i^\top\A\Q_{-i}\A^\top\x_i\big/\gamma_i \leq m/\gamma = m-d$.

Thus, putting this together we conclude that:
\begin{align*}
  |\gamma-\bar\gamma|\cdot \big(1- \tfrac{6d}{(m-d)(1-\delta)}\big) \leq
  O(\alpha d/m^2) + O(\sqrt{\alpha d}/m) = O(\sqrt{\alpha d}/m),
\end{align*}
which for $m\geq8 d$ and $\delta\leq 1/m^3$ implies that $(\gamma-\bar\gamma)^2 =
O(\alpha d/m^2)$ so we get $T_1=O(\sqrt{\alpha d}/m)$. Finally, we
obtain the bound $\|\H^{\frac12}\Z_2\H^{-\frac12}\|\leq T_1\cdot
T_2\cdot T_3= O(\alpha\sqrt d/m)$, which concludes the proof.

\subsection{Proof of Lemma \ref{l:applying-burkholder}}
\label{s:applying-burkholder}
Let $\Q_{-ij}$ denote the matrix
$(\gamma\A^\top\S_{-ij}^\top\S_{-ij}\A)^{-1}$ where $\S_{-ij}$ is the
matrix $\S_m$ without the $i$th and $j$th rows and $\gamma=\frac{m}{m-d}$.
Let $\E_{\Ec',j}[\cdot]$ be the conditional expectation with respect to
$\Ec'$ and the $\sigma$-field $\mathcal F_j$ generating the rows
$\frac1{\sqrt m}\x_1^\top \ldots, \frac1{\sqrt m}\x_j^\top$ of $\S$.
First note that
\begin{align*}
  \tr (\Q_{-i} - \E_{\Ec'} \Q_{-i}) \A^\top \A
  &= \E_{\Ec',m} [\tr \Q_{-i} \A^\top
  \A]  - \E_{\Ec',0}[ \tr \Q_{-i} \A^\top \A]\\ 
  &= \sum_{j=1}^m \left( \E_{\Ec',j}[\tr \Q_{-i} \A^\top \A] -
    \E_{\Ec',j-1}[\tr \Q_{-i} \A^\top \A]\right)
    = -\sum_{j=1}^m (\psi_j + \xi_j),
  \\
  \text{where}\quad \psi_j
  &=(\E_{\Ec',j}-\E_{\Ec',j-1})[\tr (\Q_{-ij} - \Q_{-i}) \A^\top \A]
  \\
  \text{ and }\ \xi_j&=-(\E_{\Ec',j}-\E_{\Ec',j-1})[\tr\,\Q_{-ij} \A^\top \A].
\end{align*}
 This forms a martingale difference sequence and hence falls within the scope of the Burkholder inequality \cite{burkholder1973distribution}, recalled as follows. We mention that similar martingale decomposition techniques are common in random matrix theory, see e.g., \cite{bai2010spectral}. Also, for the case $L=2$ that we will use, Burkholder inequality is nothing but the law of iterated variance.

\begin{lemma}[\cite{burkholder1973distribution}]\label{l:burkholder}
For $\{ x_j \}_{j=1}^m$ a real martingale difference sequence with
respect to the increasing $\sigma$ field $\mathcal F_j$, we have, for
$L > 1$, there exists $C_L > 0$ such that 
\[
  \E \bigg[
    \Big| \sum_{j=1}^m x_j \Big|^L
  \bigg]
  \leq C_L \cdot \E \bigg[
    \Big( \sum_{j=1}^m |x_j|^2 \Big)^{L/2}
  \bigg].
\]
\end{lemma}
Note that for each pair $i,j$,  one of $\Ec_1,\Ec_2$ is independent 
of both $\x_i$ and $\x_j$. Without loss of generality, suppose that this is
$\Ec_1$. Then, in particular, $\Ec_1$ implies that $\A\Q_{-ij}\A^\top\preceq 6\,\I$.
Thus, conditioned on $\Ec_1$, it follows that
\begin{align*}
  \tr(\Q_{-ij} - \Q_{-i}) \A^\top \A
  &=
\tr\bigg(  \frac{\frac\gamma m\Q_{-ij} \A^\top
  \x_j \x_j^\top \A \Q_{-ij}}{ 1 + \frac\gamma m\x_j^\top \A \Q_{-ij}
  \A^\top \x_j }\A^\top\A\bigg)
\\
&=
     \frac{\frac\gamma m\x_j^\top(\A\Q_{-ij}\A^\top)^2\x_j}
     {1+\frac\gamma m\x_j^\top\A\Q_{-ij}\A^\top\x_j}
\leq
     \frac{6\cdot \frac\gamma m\x_j^\top\A\Q_{-ij}\A^\top\x_j}
     {1+\frac\gamma m\x_j^\top\A\Q_{-ij}\A^\top\x_j}\leq 6,
\end{align*}
which implies that $|\psi_j|\leq 6$. We now provide a bound
on the second moment of~$\psi_j$, bounding the $\Ec'$-conditional
expectation in terms of the $\Ec_1$-conditional expectation
analogously as in \eqref{eq:conditional-bound}:
\begin{align*}
  \E_{\Ec'}[\psi_j^2]
  &\leq
  2\cdot\E_{\Ec_1}\Bigg[\bigg(\frac{6\cdot \frac\gamma m\x_j^\top\A\Q_{-ij}\A^\top\x_j}
     {1+\frac\gamma m\x_j^\top\A\Q_{-ij}\A^\top\x_j}\bigg)^2\Bigg]
\leq
    72\cdot \E_{\Ec_1}[\tfrac\gamma m\x_j^\top\A\Q_{-ij}\A^\top\x_j]
      \\
  &=
    72\cdot \frac{\E_{\Ec_1}[\tr\,\A\Q_{-ij}\A^\top]}{m-d}\leq 72\cdot
    6\cdot \frac d{m-d}.
\end{align*}
We now aim to bound
$|\xi_j|$. Since $\Ec_1$ is independent of $\x_j$, we
have $\E_{\Ec_1,j}[\tr\,\Q_{-ij}\A^\top\A] = \E_{\Ec_1,j-1}[\tr\,\Q_{-ij}\A^\top\A]$.
Furthermore, letting $\delta_2=\Pr(\neg\Ec_2)$, we have:
\begin{align*}
  \E_{\Ec_1,j-1}[\tr\,\Q_{-ij}\A^\top\A]
  &=
    \E_{\Ec',j-1}[\tr\,\Q_{-ij}\A^\top\A](1-\delta_2)+\E_{\Ec_1,j-1}[\tr\,\Q_{-ij}\A^\top\A\mid\neg\Ec_2]\delta_2,
  \\
  \E_{\Ec_1,j}[\tr\,\Q_{-ij}\A^\top\A]
  &=
    \E_{\Ec',j}[\tr\,\Q_{-ij}\A^\top\A](1-\delta_2)+\E_{\Ec_1,j}[\tr\,\Q_{-ij}\A^\top\A\mid\neg\Ec_2]\delta_2.
\end{align*}
Thus, subtracting the two equalities from each other, we conclude that:
\begin{align*}
  |\xi_j|
  &=|(\E_{\Ec',j}-\E_{\Ec',j-1})[\tr\,\Q_{-ij} \A^\top \A]|
\\
  &\leq \delta_2\cdot\frac{|(\E_{\Ec_1,j}-\E_{\Ec_1,j-1})[\tr\,\Q_{-ij} \A^\top \A \mid\neg\Ec_2]|}
    {1-\delta_2}
  \\
  &\leq 2\delta_2\cdot 6d \leq 12\cdot d/m,\quad\text{for}\quad\delta_2\leq 1/m.
\end{align*}
So, with $x_j = \psi_j+\xi_j$ and $X =- \tr (\Q_{-i} - \E_{\Ec'}
[\Q_{-i}]) \A^\top \A$  in Lemma~\ref{l:burkholder}, for $L =2$ we get:
\begin{align*}
  \E_{\Ec'}[X^2]
  &\leq C_2 \cdot \sum_j \E_{\Ec'}\big[(\psi_j+\xi_j)^2\big]
\\
  &=
  C_2\cdot \sum_j \Big(\E_{\Ec'}[\psi_j^2] +
  2\,\E_{\Ec'}[\psi_j\xi_j]+\E_{\Ec'}[\xi_j^2]\Big)
\\
  &\leq C_2 m\cdot
  \Big(72\cdot 6\cdot\frac{d}{m-d}+2\cdot6\cdot12\cdot\frac{d}{m}+12^2\frac{d^2}{m^2}\Big)
  \leq Cd,
\end{align*}
where we also used that $m\geq 8d$, thus concluding the proof.

\section{Restricted Bai-Silverstein inequality}
\label{s:bai-silverstein}

In this section, we prove Theorem
\ref{t:generalized-bai-silverstein}. Specifically, we study
Condition \ref{cond2} (Restricted Bai-Silverstein), the second structural condition
for small inversion bias in Theorem 
\ref{t:structural}, which describes the deviation of a quadratic form
$\x^\top\B\x$ from its mean, for a random vector $\x$. We start by
showing Theorem~\ref{t:generalized-bai-silverstein}, a generalized
version of the lemma of Bai and Silverstein (Lemma
\ref{l:classical-bai-silverstein}), which applies when $\x$ 
has independent entries. Then, in Appendix \ref{s:less-bai-silverstein}
we show a similar result for a
leverage score sparsified vector, constructed as in
Definition~\ref{d:leverage-score-sparsifier}, which has
non-independent 
entries. Finally, in Appendix~\ref{s:lower-bai-silverstein} 
we consider random vectors used in other fast sketching methods,
and give lower bounds demonstrating why these methods do not provide
satisfactory guarantees for Condition \ref{cond2}.

\subsection{Proof of Theorem \ref{t:generalized-bai-silverstein}}

Since the assumptions on $\x$ only depend on the leverage scores of $\A$, and the conclusion is about the variance of a quadratic form, which only depends on the first four moments of the entries of $\x$, we can assume without loss of generality that the distribution of $\x$ only depends on the leverage scores of $\A$. We will prove the claim for such random vectors $\x$.

We start by proving the following result:

\begin{proposition}Let $\A$ be a fixed $n \times d$ matrix with $n\geq
  d$, and $\x$ be a random vector with independent entries with mean
  zero and unit variance, whose distribution only depends on the
  leverage scores of $\A$. Then, Condition \ref{cond2} (Restricted
  Bai-Silverstein) for the matrix $\A$ is equivalent to 
\begin{align*}
\lambda_{\max}\left((\U \circ \U)^\top \D (\U \circ \U)\right) \leq \alpha-2,
\end{align*}
where $\U$ is the $n \times d$ matrix of left singular vectors of $\A$ and $\D$ is the $n \times n$ matrix $\D = \diag (d_k)$, with $d_k = \E x_k^4 - 3$.
\end{proposition}

\begin{proof}
The Restricted Bai-Silverstein condition is equivalent to having, for all matrices $\B$ of the form $\B = \U \M \U^\top$, where $\U$ is the matrix of left singular vectors of $\A$,
$$\Var[\x^\top \B \x] \leq \alpha \cdot \tr(\B^2).$$
Let $\z = \U ^\top \x$. Then this is equivalent to
\begin{align}\label{c1}
\Var[\z^\top \M \z] \leq \alpha \cdot \tr(\M^2).
\end{align}

First we claim that it is enough to consider diagonal matrices $\M$.
Suppose that we have a condition $C(\diag(\U \U^\top),\alpha)$ that guarantees that \eqref{c1} holds for diagonal matrices $\M_d$, and that depends only on the leverage scores and $\alpha$. 
Now, consider a general matrix $\M$, and suppose it has the eigendecomposition $\M = \O\M_d \O^\top$ for a diagonal matrix $\M_d$. We can write the equivalences
\begin{align*}
\Var[\z^\top \M \z] &\leq \alpha \cdot \tr(\M^2)\\
\Var[\z^\top  \O\M_d \O^\top \z] &\leq \alpha \cdot \tr([\O \M_d \O^\top]^2)\\
\Var[\z_d^\top \M_d \z_d] &\leq \alpha \cdot \tr(\M_d^2)
\end{align*}
where $\z_d = \O^\top \z = (\U \O)^\top \x$. Now we apply the condition $C(\diag(\U _o \U_o^\top),\alpha)$ to $\U_o = \U \O$ and the diagonal matrix $\M_d$. This condition is applicable, because $\M_o$ is a diagonal matrix, and guarantees \eqref{c1} for $\M_d$. However, we also have that the row norms of $\U_o = \U \O$ are the same as the row norms of $\U$, because $\O$ simply acts by an orthogonal rotation of the rows. So $\diag(\U _o \U_o^\top) = \diag(\U  \U^\top)$. Thus, since the distributions of the sketches we consider only depend on the leverage scores of $\A$, which are the diagonals of the matrix $\A (\A^\top \A)^{-1} \A^\top = \U  \U^\top$, the condition $C(\diag(\U  \U^\top),\alpha)$ guarantees that \eqref{c1} holds for the original matrix $\M$. This shows that it is enough to establish the condition for diagonal matrices $\M$.

Hence we can rotate $\U$ by the eigenvectors $\O$ of $\M$ into $\U' = \U  \O$, and thus assume without loss of generality that $\M$ is diagonal, $\M = \diag(\g)$, where $\g$ is a vector. Then, the condition simplifies to
\begin{align*}
\Var[\z^\top \M \z]& = \Var[\sum_{i=1}^d z_i^2 g_i]\\
& = \g^\top \mG \g \leq \alpha \cdot \|\g\|^2,
\end{align*}
where $\mG $ is the covariance matrix of $\z\circ
\z$. Here the  symbol $\circ$ means entrywise product.
This condition has to be true for any vector $\g$. Thus, this condition says exactly that the largest eigenvalue of $\mG$ is at most $\alpha$:
\begin{align*}
\lambda_{\max}(\mG) \leq \alpha.
\end{align*}
Also we assume that $\E \x \x^\top = \I_m$, hence for any symmetric matrix $\F$ (see e.g., \cite{bai2010spectral,couillet2011random} and \cite[Lemma B.6.]{mei2019generalization}),
\begin{equation}
  \Var[\x^\top \F \x] = \sum_{k} d_k F_{kk}^2 + 2 \tr (\F^2) 
\end{equation}

where $d_k = \E x_k^4 - 3$. Therefore, applying this for $\F = \U  \diag(\g) \U^\top$, and matching terms, one has $\mG = (\U \circ \U)^\top \D \U\circ \U + 2 \I_n$, where $\D = \diag (d_k)$ and with $d_k = \E x_k^4 - 3$. This finishes the proof.
\end{proof}

We now continue with the proof of the main claim
(Theorem~\ref{t:generalized-bai-silverstein}). Based on the above
results, as 
long as the random vector $\x$ has independent entries of 
zero mean and unit variance, proving Condition \ref{cond2} boils down
to the control of the fourth moment of the 
distribution. 

Let $\bbR = \U \circ \U$, and let $\bbr_i$ denote its rows. Note that $\bbr_i$ have non-negative entries. Let $\L = \diag(1/\|\u_i\|^2) = \diag(1/\|\bbr_i\|_1) = \diag(1/l_i)$ be the matrix of inverse leverage scores of $\A$, which are also the inverse $\ell_1$ norms of the rows $\bbr_i$ of $\bbR$. We can simply discard the zero rows to ensure that this is well defined and $\|\bbr_i\|_1>0$ for all indices.

Then if we can bound $\lambda_{\max}\left(\bbR^\top \L \bbR\right) \leq \kappa$, it follows that $\lambda_{\max}\left(\bbR^\top \D \bbR\right) \leq C\kappa \leq \alpha-2$, which is our desired condition as long as $\alpha$ is sufficiently large. We will show this bound with $\kappa=1$.

Note that $\Q=\bbR^\top \L \bbR$ is a symmetric matrix and has non-negative entries, because 
the rows of $\bbR$, $\bbr_i =
\u_i \circ \u_i$ are the entry-wise squares of certain vectors, and the entries of $\L$ are all positive.
Moreover,  it is readily verified that the all ones vector $\one_d$ (which clearly has non-negative entries), is an eigenvalue of $\Q$ with unit eigenvalue, 
\begin{align*}
\Q\one_d = \one_d.
\end{align*}
In other words, $\Q$ is a symmetric doubly stochastic matrix.
In more detail, we have
\begin{align*}
\Q\one_d &= \bbR^\top \L \bbR \one_d 
= \sum_{i=1}^n \frac{\bbr_i\bbr_i^\top}{\|\bbr_i\|_1} \one_d 
= \sum_{i=1}^n \bbr_i\cdot\frac{\bbr_i^\top \one_d}{\|\bbr_i\|_1}.
\end{align*}
Now, clearly, since $\bbr_i$ have non-negative entries, we have $\bbr_i^\top \one_d = \|\bbr_i\|_1$. Therefore, we find
\begin{align*}
\Q\one_d & 
= \sum_{i=1}^n \bbr_i\cdot\frac{\|\bbr_i\|_1}{\|\bbr_i\|_1}
= \sum_{i=1}^n \bbr_i =  \one_d.
\end{align*}
In the last equality, we have used that, since the columns of $\U$ are orthogonal vectors, we
have that $\sum_{i=1}^n r_{ij} = 1$ for all $j=1,\ldots,d$.

Hence, the largest eigenvalue of $\Q$ is at least $1$. By the Perron-Frobenius theorem for non-negative matrices, it follows that the largest eigenvalue of $\Q$ is paired with an eigenvector $\v$ of non-negative entries, see e.g., \cite[claims~8.3.1~and~8.3.2]{meyer2000matrix}. 
We can write, for any such vector $\v \geq 0$, that
\begin{align*}
\Q\v &= \bbR^\top \L \bbR \v 
= \sum_{i=1}^n \frac{\bbr_i\bbr_i^\top}{\|\bbr_i\|_1} \v 
= \sum_{i=1}^n \bbr_i\cdot\frac{\bbr_i^\top \v}{\|\bbr_i\|_1}.
\end{align*}
Now, clearly $\bbr_i^\top \v/\|\bbr_i\|_1 \leq \|\v\|_\infty$. Since
each entry of each $\bbr_i$ is non-negative, we have that $0\leq (\Q\v)_j
\leq (\sum_{i=1}^n r_{ij}) \|\v\|_\infty$. 
As mentioned, we also have that $\sum_{i=1}^n r_{ij} = 1$. Hence,
\begin{align*}
0 \leq (\Q\v)_j &\leq \|\v\|_\infty, \,\, j=1,\ldots, n.
\end{align*}
Suppose $\v$ is an eigenvector of $\Q$ with eigenvalue $\lambda\geq0$, i.e., $\Q \v = \lambda \v$. Based on the above inequality, we find $\|\lambda \v\|_\infty \leq \|\v\|_\infty$, hence $\lambda\leq 1$.
This shows that the largest eigenvalue of $\Q$ is at most unity.
Thus, by the above reasoning $\lambda_{\max}\left(\bbR^\top \D
  \bbR\right) \leq C$, and thus Condition \ref{cond2} holds as long as $C+2\leq \alpha$. This finishes the proof.

\subsection{Restricted Bai-Silverstein for \less\ embeddings}
\label{s:less-bai-silverstein}
In this section we show that a sub-gaussian random vector sparsified
using our leverage score sparsifier (\less) satisfies Condition
\ref{cond2} (Restricted Bai-Silverstein) with $\alpha=O(1)$. We use
this fact in Appendix~\ref{s:subspace-embedding} to prove Theorem \ref{t:lsse}.
\begin{lemma}[Restricted Bai-Silverstein for \less]\label{l:restricted-bai-silverstein}
Fix a matrix $\A\in\R^{n\times d}$ with rank $d$ and let $\xib$ be a
leverage score sparsifier for $\A$. For any p.s.d.\ matrix $\B$
restricted to the span of $\A$ and any 
$\x^\top=(x_1,...,x_n)$ having independent entries with mean zero,
unit variance and $\E[\x_i^4]=O(1)$,
\begin{align*}
  \Var\big[(\x\circ\xib)^\top\B (\x\circ\xib)\big] \leq O(1)\cdot \tr(\B^2).
\end{align*}
\end{lemma}
\begin{proof}
Let $\U=\A(\A^\top\A)^{-1/2}$ be the orthonormal basis matrix for the
span of $\A$, and let $\U_{\xib}=\diag(\xib)\U$. Note that $\B=\U\U^\top\B\U\U^\top=\U\C\U^\top$ for
$\C=\U^\top\B\U$. It follows
that:
\begin{align*}
  \Var\big[(\x\circ\xib)^\top\B(\x\circ\xib)\big]
  &= \Var[\x^\top\U_{\xib}\C\U_{\xib}^\top\x]
    =\Var\big[\tr(\U_{\xib}\C\U_{\xib}^\top)\big]
    + \E\big[\Var_{\xib}[\x^\top\U_{\xib}\C\U_{\xib}^\top\x]\big],
\end{align*}
where $\Var_{\xib}$ denotes the conditional variance with respect to
$\xib$. Recall that $\xi_i=\sqrt{\frac{b_i}{dp_i}}$, where
  $b_i=\sum_{t=1}^d1_{[s_t=i]}$, with $s_t$ sampled i.i.d. from $(p_1,...,p_n)$ and
  $p_i\approx_{O(1)} \|\u_i\|^2/d$ (here, $\u_i^\top$ denotes the $i$th
  row of $\U$). Thus,
  $\U_{\xib}^\top\U_{\xib}=\sum_{t=1}^d\frac{\u_{s_t}\u_{s_t}^\top}{d
    p_{s_t}}$ and it follows that:
  \begin{align*}
    \Var\big[\tr(\U_{\xib}\C\U_{\xib}^\top)\big]
    &=\Var\bigg[\sum_{t=1}^d\frac{\u_{s_t}^\top\C\u_{s_t}}{dp_{s_t}}\bigg]
      =d\,\Var\bigg[\frac{\u_{s_1}^\top\C\u_{s_1}}{dp_{s_1}}\bigg]
      \\
      &\leq d\,\E\bigg[\frac{\tr(\C\u_{s_1}\u_{s_1}^\top\C\u_{s_1}\u_{s_1}^\top)}{d^2p_{s_1}^2}\bigg]
        \leq\E\bigg[\frac{\|\u_{s_1}\|^2}{dp_{s_1}}\,\frac{\u_{s_1}^\top\C^2\u_{s_1}}{p_{s_1}}\bigg]
    \\
    &\leq O(1)\,\E\bigg[\frac{\u_{s_1}^\top\C^2\u_{s_1}}{p_{s_1}}\bigg] = O(1)\,\tr(\U\C^2\U^\top)=O(1)\,\tr(\B^2).
  \end{align*}
The Bai-Silverstein inequality (Lemma \ref{l:classical-bai-silverstein}) implies that
$\Var_{\xib}[\x^\top\U_{\xib}\C\U_{\xib}^\top\x]\leq O(1)\cdot
\tr\big((\U_{\xib}\C\U_{\xib}^\top)^2\big)$, so we have: 
\begin{align*}
  \E\big[\Var_{\xib}[\x^\top\U_{\xib}\C\U_{\xib}^\top\x]\big]
  &\leq O(1)\cdot\E\Big[\tr\big((\U_{\xib}\C\U_{\xib}^\top)^2\big)\Big]
   =O(1)\cdot
    \E\bigg[\tr\bigg(\Big(\sum_{t=1}^d\frac{\C\u_{s_t}\u_{s_t}^\top}{dp_{s_t}}\Big)^2\bigg)\bigg]
  \\
  &\leq O(1)\sum_{t=1}^d
    \E\bigg[\frac{\tr(\C\u_{s_t}\u_{s_t}^\top \C\u_{s_t}\u_{s_t}^\top)}{d^2p_{s_t}^2}\bigg]
    +O(1)\sum_{t\neq r}\E\bigg[\frac{\tr(\C\u_{s_t}\u_{s_t}^\top\C
    \u_{s_r} \u_{s_r}^\top)}{dp_{s_t}\cdot dp_{s_r}}\bigg]
  \\
  &\leq O(1)\,\tr(\B^2) +
    O(1)\,\tr\Big(\C\,\E\Big[\frac{\u_{s_1}\u_{s_1}^\top}{p_{s_1}}\Big]\C\,
    \E\Big[\frac{\u_{s_2}\u_{s_2}^\top}{p_{s_2}}\Big]\Big)\leq O(1)\cdot\tr(\B^2).
\end{align*}
Thus, we obtain the desired bound:
\begin{align*}
\Var\big[(\x\circ\xib)^\top\B(\x\circ\xib)\big]\leq O(1)\cdot \tr(\B^2),
\end{align*}
which completes the proof.
\end{proof}

\subsection{Lower bounds for other sketching methods}
\label{s:lower-bai-silverstein}

In this section, we show lower bounds for Condition \ref{cond2} (Restricted Bai-Silverstein) in
the context of existing fast sketching techniques. To do that, we first discuss the basic requirement of the
framework defined by Theorem \ref{t:structural}, namely that the
sketching matrix $\S$ must have i.i.d.\ rows.
  
\textbf{Fast sketches with i.i.d.\ rows.}
In our discussion, we will focus on three types of sketches:
approximate leverage score sampling \cite{drineas2006sampling},
Subsampled Randomized Hadamard 
Transform \cite{ailon2009fast}, and sparse embedding matrices
(extensions of the CountSketch \cite{cw-sparse}, see \cite{nn-sparse,cohen2016nearly}),
all of which can be implemented in time nearly linear in the input size.
The i.i.d.\ row assumption can be easily satisfied
by any row sampling sketch, including approximate leverage score
sampling. The SRHT technically does not satisfy this 
assumption, however if we treat the Randomized Hadamard Transform as a
preprocessing step (given that it does not distort the covariance matrix),
then the subsampling part can be analyzed analogously as leverage score
sampling. In the case of sparse embedding matrices, the most commonly
studied variant has a fixed number of non-zeros per column of $\S$ and
so it does not have independent rows, however, it is known that a variant
with independently sparsified entries (which fits into the setup of
Theorem \ref{t:structural}) achieves nearly matching approximation
guarantees \cite{cohen2016nearly}. 

\textbf{Leverage score sampling.} Let $\S$ be a row sampling sketch of
size $m$, i.e., each row is distributed independently as $\frac1{\sqrt m}\x^\top$, where
$\x=\frac1{\sqrt{p_{s}}}\e_{s}$ and $s$ is an index drawn from
distribution $(p_1,...,p_n)$. Given a 
matrix $\A\in\R^{n\times d}$ of rank $d$, we call this an approximate
leverage score sampling sketch if $p_i\approx_{O(1)}l_i/d$ for all $i$, where
$l_i=\a_i^\top(\A^\top\A)^{-1}\a_i$ is the $i$th leverage score
of~$\A$. We will present two lower bound constructions. 

1. \underline{Approximate sampling and arbitrary $\A$}.
Now, suppose that $n$ is even and consider the following
specific example: 
\begin{align*}
  p_{j} =
  \begin{cases}
    l_j/2d,
    &\text{for }j\leq n/2,\\[1mm]
    3l_j/2d,
    &\text{otherwise.}
  \end{cases}
\end{align*}
Further, consider the matrix $\B=\A(\A^\top\A)^{-1}\A^\top = \P$,
which is the projection onto the column-span of $\A$, and therefore
satisfies the restriction requirement in the Restricted
Bai-Silverstein condition. Then, since
$\tr(\B^2)=\tr(\P^2)=\tr(\P)=d$, we have: 
\begin{align*}
  \Var[\x^\top\B\x]
  &= \E\Big[\big(\e_s^\top\A(\A^\top\A)^{-1}\A^\top\e_s/p_s -
    d\big)^2\Big]
= \E\Big[\big(l_s/p_s - d\big)^2\Big]
\geq (d/3)^2
\end{align*}

2. \underline{Exact sampling and a specific $\A$}.
  Suppose that $\A^\top\A=\I$, each $\a_i$ is a standard
  basis vector scaled by $\sqrt{d/n}$ and we are sampling index
  $s$ according to exact leverage scores, i.e., uniformly at random. Then, letting $\x =
  \frac1{\sqrt{p_s}}\e_s$ and $\B=\A\C\A^\top$, we have:
  \begin{align*}
    \Var[\x^\top\B\x]
    & = \E\Big[\big(\x^\top\B\x-\tr(\B)\big)^2\Big]
=\E\Big[\big(d\cdot\frac{\a_s^\top\C\a_s}{\a_s^\top(\A^\top\A)^{-1}\a_s}
      - \tr(\C)\big)^2\Big]
    \\
    &=d^2\cdot\frac1d\sum_{j=1}^d\Big(C_{jj} -
      \frac1d\sum_{i=1}^dC_{ii}\Big)^2
=d^2\cdot \Omega(1),\quad\text{if}\quad
      C_{ii}=
      \begin{cases}1/2, &\text{for even $i$,}\\
        3/2,&\text{for odd $i$.}
      \end{cases}
  \end{align*}
In both constructions, we have $\tr(\B^2)=\Theta(d)$, so this implies
that for leverage score sampling, Condition \ref{cond2} can
only be shown with factor $\alpha=\Omega(d)$, as opposed to $O(1)$ for
sub-gaussian sketches and \less\ embeddings.

\textbf{Data-oblivious sparse embeddings.} Let $\S$ be a sketch of
size $m$, where each row is distributed independently as $\frac1{\sqrt
  m}\x^\top$ and $\x=(\sqrt{\frac ms}b_1r_1,...,\sqrt{\frac
  ms}b_nr_n)$, with $b_i\sim\mathrm{Bernoulli}(\frac sm)$ and $r_i$ 
  distributed as a uniformly random sign. While this is not the most
  commonly studied variant of a sparse embedding, it is known to
  satisfy the subspace embedding property for sketch size $m=O(d\log
  d)$ with sparsity level $s=O(\log d)$ \cite{cohen2016nearly}, which
 matches the state-of-the-art for sparse embeddings. Other sparse
 embeddings have non-i.i.d.\ row distributions
 \cite{cw-sparse,nn-sparse,mm-sparse}, and so they do not fit into the 
 framework laid out by Theorem \ref{t:structural}. The key difference
 between the sparsification of $\x$ relative to our \less\ embeddings
 is that it is data-oblivious. We can exploit that in our lower bound
 example by choosing an extremely skewed leverage score distribution
 of matrix $\A$. In particular, suppose that $\A^\top\A=\I$ and
 moreover, $\a_i=\e_i$ for $i=1,...,k$ (where $1\leq k\leq d$) and for
 all $i>k$, the first $k$ coordinates of $a_i$ are zero. This
 construction ensures that the first $k$ leverage scores of $\A$ are
 equal $1$. Once again setting $\B=\A(\A^\top\A)^{-1}\A^\top$, we get:
 \begin{align*}
   \Var[\x^\top\B\x] \geq \sum_{i=1}^k\Var\Big[\frac ms b_ir_i\Big] =
   k\cdot \frac ms\Big(1-\frac sm\Big).
 \end{align*}
If we let $k=\Omega(d)$, then we get $\Var[\x^\top\B\x]\geq\Omega(m/s)\cdot
\tr(\B^2)$. Thus, unless we zero-out merely a constant
fraction of entries of $\S$, the sketching matrix will not satisfy
Condition \ref{cond2} with a constant factor
$\alpha=O(1)$. We conjecture that this example can be extended to show
a general lower bound on the inversion bias, as we did for approximate
leverage score sampling.

\section{Subspace embedding guarantee for \less\ embeddings}
\label{s:subspace-embedding}
In this section, we prove Lemma \ref{l:subspace-embedding} and Theorem \ref{t:lsse}.
In particular, we prove that \less\ embeddings achieve the subspace
embedding property for a sketch of size $O(d\log d)$ (Lemma \ref{l:subspace-embedding}), thus
establishing Condition \ref{cond1}. Then,
at the end of the section we briefly discuss how to combine Lemmas
\ref{l:subspace-embedding} and \ref{l:restricted-bai-silverstein},
using the structural conditions via Theorem \ref{t:structural}, to obtain Theorem \ref{t:lsse}. 

\subsection{Proof of Lemma \ref{l:subspace-embedding}}
First, note that instead of directly showing the subspace embedding of
$\S\A$ for the span of $\A$, it suffices to show the guarantee when
replacing $\A$ with its orthonormal basis matrix
$\U=\A(\A^\top\A)^{-1/2}$, since
$\A^\top\S^\top\S\A=(\A^\top\A)^{\frac12}\U^\top\S^\top\S\U(\A^\top\A)^{\frac12}$. 
Then, a standard technique, e.g., as used for leverage score sampling
sketches, relies on the following decomposition of
$\U^\top\S^\top\S\U$ as an average of independent rank-one p.s.d. random matrices:
\begin{align*}
  \U^\top\S^\top\S\U = \sum_{i=1}^m\U^\top\s_i\s_i^\top\U,
\end{align*}
where $\s_i^\top$ represents the $i$th row of $\S$.
For standard leverage score sampling sketches it suffices to use
the matrix Chernoff bound \cite[Theorem~1.1]{tropp2012user}, which
uses an almost sure bound on each rank-one matrix to ensure
concentration around the mean, $\E[\U^\top\S^\top\S\U]=\I$. However, in the case of
a leverage score sparsified 
embedding an almost sure bound is not sufficient. Instead, we show
that the rank-one matrices $\U^\top\s_i\s_i^\top\U$ exhibit
sub-exponential tails on all of their moments, as required by the
following variant of the matrix Bernstein bound.
\begin{lemma}[{\cite[Theorem~6.2]{tropp2012user}}]\label{theo:sub-exp-Bernstein}
For $i=1,2,...$, consider a finite sequence $\X_i$ of $d\times d$
independent and symmetric random matrices such that   
\[
	\E[\X_i] = \mathbf{0}, \quad \E[\X_i^p] \preceq \frac{p!}2
        \cdot R^{p-2} \A_i^2\quad\textmd{for}\quad p=2,3,...
\]
Then, defining the variance parameter $\sigma^2 = \| \sum_i
\A_i^2 \|$, for any $t>0$ we have:
\begin{align*}
	\Pr \bigg\{ \lambda_{\max}\Big( \sum\nolimits_i \X_i \Big) \geq t
  \bigg\}& \leq d \cdot \exp \left( \frac{ -t^2/2 }{ \sigma^2 +
            R t } \right).
\end{align*}
\end{lemma}

\noindent
We apply the above result for
$\X_i=\pm(\U^\top\s_i\s_i^\top\U-\frac1m\I)$, where
$\s_i=\frac1{\sqrt m}(\x_i\circ\xib)$ is a
leverage score sparsified sub-gaussian random vector. We next
establish the subexponential moment bound needed for the matrix
Bernstein bound.
\begin{lemma}\label{l:subexponential-moments}
Fix a matrix $\U\in\R^{n\times d}$ such that $\U^\top\U=\I$. Suppose that
$\xib$ is a leverage score sparsifier for $\U$ and $\x$ has
i.i.d. $O(1)$-sub-gaussian entries with mean zero and unit variance.
Then, there is $C=O(1)$ such that for all $p=2,3,...$ we have 
\begin{align*}
  \Big\|\E\Big[\Big(\U^\top(\x\circ\xib)(\x\circ\xib)^\top\U -
  \I\Big)^p\Big] \Big\|\leq \frac {p!}2 \cdot (Cd)^{p-1}.
\end{align*}
\end{lemma}

\noindent
Now, the matrix Bernstein bound (Lemma \ref{theo:sub-exp-Bernstein})
can be invoked with $\A_i^2=\frac {Cd}{m^2}\cdot\I$ and $\sigma^2=R=\frac {Cd}m$,
obtaining that for $\eta\in(0,1)$:
\begin{align*}
  \Pr\Big\{\big\|\U^\top\S^\top\S\U-\I\big\|\geq \eta\Big\}\leq
  2d\cdot\exp\Big(-\frac{\eta^2m}{4Cd}\Big)\leq
  \delta\quad\text{for}\quad m\geq 4Cd\log(2d/\delta)/\eta^2,
\end{align*}
which completes the proof.

\subsection{Proof of Lemma \ref{l:subexponential-moments}}
The key part of our proof of Lemma \ref{l:subexponential-moments}  involves
establishing the following concentration inequality which can be
viewed as a form of the Hanson-Wright inequality
\cite{rudelson2013hanson} that takes advantage of 
the leverage score sparsifier $\xib$, similarly as we did for the
Restricted Bai-Silverstein inequality (Lemma~\ref{l:restricted-bai-silverstein}).
\begin{lemma}\label{l:restricted-hanson-wright}
Fix a matrix $\U\in\R^{n\times d}$ such that $\U^\top\U=\I$. Suppose that
$\xib$ is a leverage score sparsifier for $\U$ and $\x$ has
independent $O (1)$-sub-gaussian entries with mean zero and unit
variance. Then, there is $c=\Omega(1)$ and $C=O(1)$ such that for any $t\geq Cd$ we have:
\begin{align*}
\Pr\Big\{(\x\circ\xib)^\top\U\U^\top(\x\circ\xib)\geq t\Big\}\leq
  \exp\Big(-c\,\big(\sqrt t + t/d\big)\Big).
\end{align*}
\end{lemma}
\begin{proof}
  We use the shorthand $\U_\xib=\diag(\xib)\U$. Similarly as for
  Lemma~\ref{l:restricted-bai-silverstein}, our strategy is to 
  show that the sparsification $\U_\xib$ preserves enough of the
  structure of $\U$ so that we can apply the classical Hanson-Wright
  inequality, which is repeated below, following \cite{rudelson2013hanson}, 
\begin{lemma}[{\cite[Theorem~1.1]{rudelson2013hanson}}]\label{l:hanson-wright}
Let $\x$ have independent $O (1)$-sub-gaussian entries with
mean zero and unit variance. Then, there is $c=\Omega(1)$ such that
for any $n\times n$ matrix $\B$ and $t\geq 0$,
\begin{align*}
  \Pr\Big\{|\x^\top\B\x-\tr(\B)|\geq t\Big\}\leq
  2\exp\bigg(-c\min\Big\{\frac{t^2}{\|\B\|_F^2},\frac{t}{\|\B\|}\Big\}\bigg).
\end{align*}
\end{lemma}
To show that the leverage score sparsification $\U_\xib$ is
sufficiently accurate, we can rely on the matrix Chernoff bound,
repeated below, and the following decomposition:
\begin{align*}
  \U_\xib^\top\U_\xib = \sum_{i=1}^d\frac{\u_{s_i}\u_{s_i}^\top}{dp_{s_i}},
\end{align*}
where $s_i$ are the random indices sampled from the approximate
leverage score distribution $(p_1,...,p_n)$ (see Definition
\ref{d:leverage-score-sparsifier}). For simplicity, we only repeat the
large deviation part of the Chernoff bound, which is the one relevant to our analysis.
\begin{lemma}[{\cite[Theorem~1.1 and Remark 5.3]{tropp2012user}}]\label{theo:matrix-Chernoff}
    For $i=1,2,...$, consider a finite sequence $\X_i$ of $d\times d$
independent positive semi-definite random matrices such that
$\E\big[\sum_i\X_i\big] = \I$ and $\|\X_i\|\leq R$. Then, for any $t\geq \ee$,
we have:
\begin{align*}
  \Pr\Big\{\Big\|\sum_i\X_i\Big\|\geq t\Big\}\leq
  d\cdot\Big(\frac\ee t\Big)^{t/R}.
\end{align*}
\end{lemma}
We apply the matrix Chernoff to
$\X_i=\frac1{dp_{s_i}}\u_{s_i}\u_{s_i}^\top$, noting that since
$p_i\geq\|\u_i\|^2/Rd$ for $R=O(1)$, it follows that $\|\X_i\|\leq R$. Moreover,
$\E[\sum_{i=1}^d\X_i]=\I$, so for $t\geq O(1)\cdot d$ we have:
\begin{align*}
  \Pr\big\{\|\U_\xib^\top\U_\xib\|\geq \sqrt t\big\}\leq
  d\exp\big(-\sqrt t\ln(\sqrt t/\ee)/R\big)\leq \exp(-c\sqrt t),
\end{align*}
for some $c=\Omega(1)$. Also, note that
$\|\U_\xib^\top\U_\xib\|\leq\tr(\U_\xib^\top\U_\xib)\leq Rd$ almost surely, which
implies that event $\Ec: \big[\|\U_\xib^\top\U_\xib\|\leq \min\{\sqrt
t,Rd\}\big]$ holds with probability at least $1-\exp(-c(\sqrt t +
t/d))$. Conditioned on $\Ec$, it holds that
$\|\U_\xib\U_\xib^\top\|_F^2\leq \|\U_\xib^\top\U_\xib\|\cdot
\tr(\U_\xib^\top\U_\xib)\leq \min\{\sqrt t,Rd\}\cdot Rd$, so applying
Lemma \ref{l:hanson-wright} for fixed $\xib$ we get:
\begin{align*}
  \Pr\big\{\x^\top\U_\xib\U_\xib^\top\x \geq Rd+t\ \mid\ \xib,\Ec\big\}
  &\leq
    2\exp\bigg(-c\min\Big\{\frac{t^2}{\|\U_\xib\U_\xib^\top\|_F^2},\frac{t}{\|\U_\xib\U_\xib\|}\Big\}\bigg)
  \\
  &\leq     2\exp\bigg(-c\min\Big\{\frac{t^2}{\min\{\sqrt t,Rd\}\cdot
    {Rd} },\frac{t}{\min\{\sqrt t,Rd\}}\Big\}\bigg)
  \\
  &\leq 2\exp\big(-c(\sqrt t +t/Rd)\big).
\end{align*}
Appropriately rescaling $t$, we obtain the claim.
\end{proof}

We are now ready to present the proof of Lemma
\ref{l:subexponential-moments}, obtaining subexponential moment bounds
for the random matrix $\U^\top(\x\circ\xib)(\x\circ\xib)^\top\U$, thus
completing the proof of the subspace embedding guarantee for leverage
score sparsified sketches.
\begin{proof}[Proof of Lemma \ref{l:subexponential-moments}] 
Throughout the proof, we will use the shorthand
$\U_{\xib}=\diag(\xib)\U$.  It is easy to show by induction over $p$
that: 
  \begin{align*}
\underbrace{\big(\U_{\xib}^\top\x\x^\top\U_{\xib}-\I\big)^p}_{\Z^p} =
    \big(\x^\top\U_\xib\U_\xib^\top\x-1\big)^{p-1}\U_\xib^\top\x\x^\top\U_\xib
    - \underbrace{\big(\U_{\xib}^\top\x\x^\top\U_{\xib}-\I\big)^{p-1}}_{\Z^{p-1}}.
  \end{align*}
Thus, it follows that for any $p=2,3,...$ (both even and odd) we have
the following upper bound:
\begin{align*}
  \big\|\E[\Z^p]\big\|
  \leq
  \Big\|\E\big[\underbrace{|\x^\top\U_\xib\U_\xib^\top\x-1|^{p-1}\U_\xib^\top\x\x^\top\U_\xib}_{\T_p}
  \big]\Big\|
  +  \big\|\E[\Z^{p-1}]\big\|.
\end{align*}
To bound the quadratic form $\x^\top\U_\xib\U_\xib^\top\x$ in the
first term, we can use
Lemma \ref{l:restricted-hanson-wright}. In particular, the 
lemma implies that the event $\Ec:[\x^\top\U_\xib\U_\xib^\top\x\leq
Cpd]$ fails with probability at most $\ee^{-\sqrt{pd}}$
for a sufficiently large $C=O(1)$, so we have:
\begin{align*}
  \big\|\E[\T_p]\big\|
  &\leq   \big\|\E[\T_p\cdot\one_\Ec]\big\| +
  \big\|\E[\T_p\cdot\one_{\neg\Ec}]\big\|
  \\
  &\leq (pd)^{p-1}\big\|\E[\U_\xib^\top\x\x^\top\U_\xib]\big\| +
    \E\big[(\x^\top\U_\xib\U_\xib^\top\x\cdot\one_{\neg \Ec})^p\big]
  \\
  &= (Cpd)^{p-1} +
    \int_0^\infty
    pt^{p-1}\Pr\big\{\x^\top\U_\xib\U_\xib^\top\x\cdot\one_{\neg
    \Ec}>t\big\} dt
  \\
  &\leq (Cpd)^{p-1} + p(Cpd)^p\ee^{-\sqrt{pd}} +  \int_{Cpd}^\infty
    pt^{p-1}\ee^{-c(\sqrt t+t/d)} dt.
\end{align*}
Note that $(O(1)\,p)^{p+O(1)}d^{p-1}\leq p^p(O(1)\, d)^{p-1}\leq
(p!/2)(O(1)\,d)^{p-1}$, and also $\ee^{-\sqrt{pd}}\leq O(1/d)$, so the first two terms can be easily
bounded as desired. To bound the last term, we use the following
integral formula:
\begin{align*}
  \int t^{p-1}\ee^{-\alpha t^\theta}dt =
  -\frac{\Gamma(p/\theta,\alpha t^\theta)}{\theta \alpha^{p/\theta}} + \mathrm{const},
\end{align*}
which follows from the definition of the upper incomplete Gamma
function $\Gamma$. Note that for $p=2,3,...$ this function also satisfies:
\begin{align*}
  \Gamma(p,\lambda)
  &= (p-1)!\cdot \Pr\{x<p\}
\qquad\text{for}\quad x\sim\mathrm{Poisson}(\lambda),
  \\
  & \leq (p-1)!\cdot\ee^{-c\lambda}
    \qquad\text{for}\quad\lambda\geq 2p,
   \ c=\Omega(1),
\end{align*}
where the last inequality is a standard tail bound for a Poisson random variable.
With a slight abuse of notation, we let $c$ denote the minimum of the
above constant $c$ and the constant $c$ from Lemma
\ref{l:restricted-hanson-wright}. We apply the integral formula in two different ways, depending
on $p$. First, if $p<d$ then we have:
\begin{align*}
  \int_{Cpd}^\infty  pt^{p-1}\ee^{-c(\sqrt t+t/d)} dt
  &\leq   \int_{Cpd}^\infty  pt^{p-1}\ee^{-c\sqrt t} dt
    =2pc^{-2p}\Gamma(2p,c\sqrt {Cpd})
\leq 2c^{-2p}(2p)!\ee^{-c^2\sqrt{Cpd}}.
\end{align*}
By using the fact that $\exp(-c^2\sqrt{Cpd})\leq \exp(-c^2p)=O(1/p)$,
this expression can be bounded by $(p!/2)(O(1)\,p)^{p-1}\leq
(p!/2)(O(1)\,d)^{p-1}$. Next, we consider the case when $p\geq d$. We have:
\begin{align*}
  \int_{Cpd}^\infty  pt^{p-1}\ee^{-c(\sqrt t+t/d)} dt
  &\leq  \int_{Cpd}^\infty  pt^{p-1}\ee^{-ct/d} dt
    =p(d/c)^{p}\Gamma(p,cCp) \leq p!d^p\ee^{-c^2Cp},
\end{align*}
where the last inequality holds as long as $C\geq 2/c$. Here, we
note that $\ee^{-c^2Cp}\leq O(1/d)$ since $p\geq d$, thus again
obtaining a bound of the form $(p!/2)(O(1)\,d)^{p-1}$. Putting
everything together, we conclude that:
\begin{align*}
  \|\E[\Z^p]\|\leq \frac {p!}2\,(O(1)\,d)^{p-1} + \|\E[\Z^{p-1}]\|.
\end{align*}
Recursively summing up this bound concludes the proof.
\end{proof}
\subsection{Proof of Theorem \ref{t:lsse}}
In Lemma \ref{l:subspace-embedding}, we showed that a \less\ embedding of size $m\geq 
4Cd\log(3d/\delta)$ satisfies Condition~\ref{cond1} (subspace
embedding) for $\eta=1/2$ with probability $1-\delta/3$, as required
by Theorem~\ref{t:structural}. Also, in 
Lemma~\ref{l:restricted-bai-silverstein} we showed  
Condition \ref{cond2} (Restricted Bai-Silverstein) with $\alpha=O(1)$
for a leverage score sparsified sub-gaussian vector. Thus, as long
as $\delta\leq 1/m^3$ and $m/3\geq 
4Cd\log(3d/\delta)$, it follows that
$(\frac{m}{m-d}\A^\top\S^\top\S\A)^{-1}$ is an
$(\epsilon,\delta)$-unbiased estimator of $(\A^\top\A)^{-1}$ for
$\epsilon=O(\sqrt d/m)$, and we obtain the desired guarantee. Note
that the condition for invoking Theorem \ref{t:structural} can be
written as $m\geq C'd\log(m)$. This is satisfied for
$m=C'd\log(C'^2d^2)$, and since $m$ grows faster than $\log(m)$, it
will also be satisfied for all
$m\geq C'd\log(C'^2d^2)=O(d\log(d))$. This completes the proof of Theorem \ref{t:lsse}.

\section{Averaging nearly-unbiased estimators}
\label{s:averaging}

In this section, we show that averaging improves spectral
approximation for matrix estimators with small inversion bias, and as
a consequence we prove Corollaries \ref{c:subgaussian} and
\ref{c:less} for averaging sketched inverse covariance matrix estimators based
on sub-gaussian sketches and \less\ embeddings respectively.

\subsection{Conditions for effective averaging of random matrices}

We start with a more general result, which should be of interest
to averaging nearly-unbiased matrix estimators in settings other than
inverse covariance matrix estimation.

  \begin{lemma}[Conditions for effective averaging]\label{t:averaging}
    Suppose that $\delta\leq\epsinv\leq\epssub\leq 1$ and
    $\tilde\C_1,...,\tilde\C_q$ are i.i.d.~positive semi-definite $d$-dimensional 
    random matrices such that:
    \begin{enumerate}
    \item $\tilde\C_i$ is an $(\epsinv,\delta/2q)$-unbiased estimator of $\C$;
    \item $\tilde\C_i$ is an $(\epssub,\delta/2q)$-approximation of $\C$.
    \end{enumerate}
Then, $\frac1q\sum_{i=1}^q\tilde\C_i$ is an
$(\epsilon',2\delta)$-approximation of $\C$ for
$  \epsilon' =
    \epsinv  + \epssub\cdot O\big(\sqrt
  {\frac{\ln(d/\delta)}{q}}\,\big)$.
  \end{lemma}
\begin{proof}
For this, we use a variant of the matrix Bernstein
inequality given below.
\begin{lemma}[{\cite[Theorem~1.4]{tropp2012user}}]\label{theo:Bernstein}
For $i=1,2,...$, consider a finite sequence $\X_i$ of $d\times d$
independent and symmetric random matrices such that   
\[
	\E[\X_i] = \mathbf{0}, \quad \lambda_{\max}(\X_i)\leq
        R\quad\text{almost surely.}
\]
Then, defining the variance parameter $\sigma^2 = \| \sum_i\E[\X_i^2] \|$, for any $t>0$ we have:
\begin{align*}
	\Pr \bigg\{ \lambda_{\max}\Big( \sum\nolimits_i \X_i \Big) \geq t
  \bigg\}& \leq d \cdot \exp \left( \frac{ -t^2/2 }{ \sigma^2 +
            R t/3 } \right).
\end{align*}
\end{lemma}

\noindent
Suppose that $\tilde\C$ is an
$(\epsinv,\delta/2q)$-unbiased estimator and an
$(\epssub,\delta/2q)$-spectral approximation for $\C$, with $\Ec_{inv}$
and $\Ec_{sub}$ the associated high probability events. For
concreteness, let the $O(1)$ constant factor in Definition
\ref{d:unbiased-estimator} be denoted as $M$. Further, let
$\tilde\C_i,\Ec_{inv}^i,\Ec_{sub}^i$ be the i.i.d.~copies of
$\tilde\C$ with their associated events. Finally, let $\tilde\C_i'$ be
a random matrix obtained from conditioning $\tilde\C_i$ on
$\Ec_{inv}^i\wedge\Ec_{sub}^i$, and coupled with $\tilde\C_i$ so that
$\Pr(\tilde\C_i'=\tilde\C_i)\geq \Pr(\Ec_{inv}^i\wedge\Ec_{sub}^i)\geq
1-\delta/q$ (this coupling can be obtained by considering a
construction of $\tilde\C_i'$ via rejection sampling from
$\tilde\C_i$). We can bound the bias of $\tilde\C_i'$ (for any $i$) by observing that:
\begin{align*}
-\delta/q\cdot \E[\tilde\C_i\mid\Ec_{inv}^i,\neg\Ec_{sub}^i]\preceq  \E[\tilde\C_i']- \E[\tilde\C_i\mid\Ec_{inv}^i]\preceq
  \frac{\delta/q}{1-\delta/q}\cdot \E[\tilde\C_i\mid \Ec_{inv}^i].
\end{align*}
Since we have $\E[\tilde\C_i\mid \Ec_{inv}^i]\approx_{\epsinv}\C$ and
$\E[\tilde\C_i\mid\Ec_{inv}^i,\neg\Ec_{sub}^i]\preceq M\cdot \C$, it
follows that $\E[\tilde\C_i']$ is an $\epsilon'$-spectral approximation
of $\C$ for $\epsilon'=\epsinv +
\frac{2\delta}q(1+\epsinv+M)$.

We will now apply the matrix Bernstein inequality (Lemma
\ref{theo:Bernstein}) to the sequence of matrices:
\[\X_i=\frac1q\Big(\C^{-\frac12}\tilde\C_i'\C^{-\frac12}-\E\big[\C^{-\frac12}\tilde\C_i'\C^{-\frac12}\big]\Big),
  \quad i=1,...,q.\]
Note that we have $\tilde\C_i'\approx_{\epssub}\frac1q\C$, so it
follows that $\|\X_i\|\leq (\epssub+\epsilon')/q$ and
$\sum_i\|\X_i^2\|\leq (\epssub+\epsilon')^2/q$. Thus, we conclude
that for $t\in(0,1)$:
\begin{align*}
  \Pr\bigg\{\Big\|\sum_{i=1}^q\X_i\Big\|\geq
  t\,(\epssub+\epsilon')\bigg\}\leq 2d\exp\big(-t^2q/4\big).
\end{align*}
Setting $t=\sqrt {4\ln(2d/\delta)/q}$ (without loss of generality,
assume that $t\leq 1$), we obtain that with probability $1-\delta$,
\begin{align*}
  \Big\|\frac1q\sum_{i=1}^q\C^{-\frac12}\tilde\C_i'\C^{-\frac12} -
  \I\Big\|
  &\leq \Big\|\sum_{i=1}^q\X_i\Big\| +
    \Big\|\frac1q\sum_{i=1}^q\C^{-\frac12}\E[\tilde\C_i']\C^{-\frac12} -
    \I\Big\|
  \\
  &\leq t\cdot(\epssub+\epsilon') + \epsilon'
  \leq \epsinv + \epssub\cdot
    O\Big(\sqrt{\tfrac{\log(d/\delta)}{q}}\,\Big) +  O\Big(\frac{\delta M}{q}\Big).
\end{align*}
Note that under the assumptions that $M =O(1)$ and $\delta\leq
\epssub$, we can absorb the last term into the middle term. Finally,
observe that thanks to the coupling and a union bound, the above bound holds with
probability $1-2\delta$ after we replace $\tilde\C_i'$ with
$\tilde\C_i$, completing the proof of Lemma~\ref{t:averaging}.
\end{proof}

\subsection{Proof of Corollary \ref{c:subgaussian}}
Consider a sub-gaussian sketching matrix $\S$ of size $m\geq C(d+\sqrt
d/\epsilon+\log(2q/\delta))$. From Proposition \ref{t:subgaussian}, it
follows that $(\frac m{m-d}\A^\top\S^\top\S\A)^{-1}$ is an
$(\epsilon,\delta/2q)$-unbiased estimator of $(\A^\top\A)^{-1}$. Further,
it is an $(\epssub,\delta/2q)$-approximation of $(\A^\top\A)^{-1}$,
where $\epssub = O(\sqrt{d/m})=O(\epsilon\cdot\sqrt m)$. Thus, using
Lemma~\ref{t:averaging}, it follows that for $q$ i.i.d.\ copies
$\S_1,...,\S_q$, the averaged estimator $\frac1q\sum_{i=1}^q(\frac
m{m-d}\A^\top\S_i^\top\S_i\A)^{-1}$ is an
$(\epsilon'',2\delta)$-approximation of $(\A^\top\A)^{-1}$ for
\[\epsilon'' =
  \epsilon+O\Big(\epsilon\cdot\sqrt{m\log(d/\delta)/q}\,\Big).\]
Setting
$q=O(m\log(d/\delta))$ and adjusting the constants appropriately, we
obtain the claim.

\subsection{Proof of Corollary \ref{c:less}}
Consider a \less\ embedding matrix $\S$ of size $m\geq
C(d\log(2dq/\delta) + \sqrt d/\epsilon)$. From Theorem~\ref{t:lsse},
it follows that $(\frac m{m-d}\A^\top\S^\top\S\A)^{-1}$ is an
$(\epsilon,\delta/2q)$-unbiased estimator of
$(\A^\top\A)^{-1}$. Furthermore, the theorem also implies that this
matrix is an $(\epssub,\delta/2q)$-approximation of $(\A^\top\A)^{-1}$
for $\epssub = O(\sqrt{d\log(2dq/\delta)/m}) = O(\epsilon\cdot
\sqrt{m\log(d/\delta)})$. Using Lemma~\ref{t:averaging}, it follows
that for $q$ i.i.d.\ copies $\S_1,...,\S_q$, the averaged estimator
$\frac1q\sum_{i=1}^q(\frac m{m-d}\A^\top\S_i^\top\S_i\A)^{-1}$ is an
$(\epsilon'',2\delta)$-approximation of $(\A^\top\A)^{-1}$ for
\[\epsilon'' = \epsilon+O\bigg(\epsilon\cdot\sqrt{m\log^2(2dq/\delta)/q}\,\bigg).\]
Setting
$q=O(m\log^2(d/\delta))$ and adjusting the constants appropriately, we
obtain the claim.

Note that in both Corollaries there is a slight interdependence in the conditions for $m$
and $q$. This is in general unavoidable, since
as $q$ grows large with fixed $m$, the average has to eventually
converge to the true expectation of $(\frac
m{m-d}\A^\top\S^\top\S\A)^{-1}$, which may be unbounded.

\section{Inversion bias lower bound for leverage score sampling}
\label{s:lower}
In this section, we show a lower bound on the inversion bias of
approximate leverage score sampling, proving Theorem
\ref{t:lower}. In the proof, we show a lower bound for the inverse
moment of a shifted Binomial random variable
(Lemma~\ref{l:inverse-moment}), which should be of independent
interest.

\subsection{Proof of Theorem \ref{t:lower}}
Without loss of generality, suppose that $n=2d$
(otherwise the matrix $\A$ can be padded by zeros). We can also assume
that $m\geq d$, since the other cases follow easily. Our construction is designed 
so that uniform row sampling is a $1/2$-approximation of leverage
score sampling. Let $\S$ be a uniform row sampling sketch of size $m$,
i.e., its $i$th row is $\sqrt{\!\frac nm}\,\e_{s_i}^\top$, where
$s_1,...,s_m$ are  independent uniformly random indices from
$1,...,n$. Our matrix $\A$ consists of $n=2d$ scaled standard basis
vectors such that pairs of consecutive rows are given by
$\a_{2(i-1)+1}^\top=\a_{2(i-1)+2}^\top=\frac1{\sqrt 2}\ \e_i^\top$ for $i\geq
2$, whereas the first two rows are $\a_1^\top=\frac1{\sqrt
  4}\e_1^\top$ and $\a_2^\top=\sqrt{\frac 34}\e_1^\top$:
\begin{align*}
  \A = 
\begin{bmatrix}
\frac1{\sqrt 4} &&&\\
\sqrt{\frac34}&&0&\\
&\frac1{\sqrt 2}&&\\
&\frac1{\sqrt 2}&&\\
&&\ddots&\\
&0&&\frac1{\sqrt 2}\\
&&&\frac1{\sqrt 2}
\end{bmatrix}.
\end{align*}
First, note that $\A^\top\A=\I$, and all of the squared row norms are
within $[\frac12\frac dn,\frac32\frac dn]$, so uniform sampling is
indeed a $1/2$-approximate leverage score sampling scheme.
Further, for any $\gamma>0$, the matrix
$\gamma\A^\top\S^\top\S\A$ is diagonal, and its diagonal entries are given by:
\begin{align*}
  \big[\gamma\A^\top\S^\top\S\A\big]_{ii} =
  \begin{cases}
   \frac{\gamma
      n}{m}\sum_{j=1}^m\big(\frac14\one_{[s_j=1]}+\frac34\one_{[s_j=2]}\big)
    =\frac{\gamma n}{m}\cdot\frac{x+b_1/2}{2}
    &\text{for }i=1,\\
\frac{\gamma
      n}{m}\sum_{j=1}^m\big(\frac12\one_{[s_j=2(i-1)+1]}+\frac12\one_{[s_j=2(i-1)+2]}\big)
    =\frac {\gamma n}{m}\cdot\frac {b_i}{2}
    &\text{otherwise,}
    \end{cases}
\end{align*}
where $b_i$'s are all identically (but not independently) distributed as $\binom(m,1/d)$
and $x$ is distributed, conditionally on
$b_1$, as $\binom(b_1,1/2)$.
Here $b_i$ denote the number of times  $s_j \in \{2i-1,2i\}$, while $x$ denotes the number of times $s_j=2$. Due to the symmetry of the problem, conditionally on a given value of $b_1$ (i.e., a given value of counts $s_j$ that are equal to either unity or two), each $s_j\in\{1,2\}$ is distributed uniformly over $\{1,2\}$, hence the value $x$ of counts $s_j$ that are equal to two is distributed as $\binom(b_1,1/2)$. This leads to the claimed distributional representation.

 The key idea in the construction is that
the first diagonal entry of the sketch has more variance than the
others, and thus it will also have more inversion bias. As a
result, there is no scaling $\gamma$ that will simultaneously correct
the inversion bias of the first entry and of all the other entries. To
that end, we lower bound a shifted inverse moment of the Binomial
distribution in the following lemma, potentially of independent
interest, proven at the end of this section.
\begin{lemma}\label{l:inverse-moment}
There is a universal constant $C>0$ such that for any positive integer $b$, if $x\sim\binom(b,1/2)$ then:
\begin{align*}
  \E\bigg[\frac1{x+b/2}\bigg] \geq \Big(1 +\frac
  1{Cb}\Big)\cdot\frac1{b}.
  \end{align*}
\end{lemma}

\noindent
Note that the expected inverse of $\gamma\A^\top\S^\top\S\A$
is undefined since the matrix may not be invertible. Thus, as in the
definition of an $(\epsilon,\delta)$-unbiased estimator, we must
condition on a high probability event which ensures invertibility. We
start by considering the largest such event, $\Ec^*:[\forall_i
b_i>0]$. Using the fact that, conditioned on $b_1$, the variable $x$ is
independent of $\Ec^*$, we have:
\begin{align*}
  \E\Big[\big[(\gamma\A^\top\S^\top\S\A)^{-1}\big]_{11}\mid\Ec^*\Big]
  &= \left(\frac{\gamma n}{m} \right)^{-1}\sum_{b>0}\E\Big[\frac{2}{x+b_1/2}\mid
    b_1=b\Big]\,\Pr(b_1=b\mid\Ec^*)
  \\
  &\overset{(a)}{\geq} \left(\frac{\gamma n}{m} \right)^{-1}\sum_{b>0}\Big(1+\frac1{Cb}\Big)\,\frac2b\,\Pr(b_1=b\mid\Ec^*)
  \\
  &=\sum_{b>0}\Big(1+\frac1{Cb}\Big)\,
    \E\Big[\big[(\gamma\A^\top\S^\top\S\A)^{-1}\big]_{22}\mid
    b_2=b\Big]\,\Pr(b_2=b\mid\Ec^*)
  \\
  &\geq\E\Big[\big[(\gamma\A^\top\S^\top\S\A)^{-1}\big]_{22}\mid\Ec^*\Big]
  \\
  &\quad+
    \frac1{2C}\frac dm\sum_{b=1}^{2m/d}\E\Big[\big[(\gamma\A^\top\S^\top\S\A)^{-1}\big]_{22}\mid
    b_2=b\Big]\,\Pr(b_2=b\mid\Ec^*)
  \\
  &\overset{(b)}{\geq} \Big(1 + \frac{d}{4Cm}\Big)\cdot
    \E\Big[\big[(\gamma\A^\top\S^\top\S\A)^{-1}\big]_{22}\mid\Ec^*\Big], 
\end{align*}
where in $(a)$ we used Lemma \ref{l:inverse-moment} and in $(b)$ we
observed that $\big[(\gamma\A^\top\S^\top\S\A)^{-1}\big]_{22}$ decreases
with $b_2$ and moreover, since $\E[b_2]=m/d\geq 1$, it is easy to verify
that the range $[1,2m/d]$ contains more than half of the probability
mass of $\binom(m,1/d)$.

The above derivation shows that when conditioned on $\Ec^*$, for any scaling $\gamma>0$ the
inversion bias will be at least $\Omega(d/m)$, since the estimated
matrix $(\A^\top\A)^{-1}=\I$ has the same entries on the diagonal,
whereas the expectation of the first two diagonal entries of the estimator
$(\gamma\A^\top\S^\top\S\A)^{-1}$ differs by
a factor of $1+\Omega(d/m)$. To complete the proof of Theorem
\ref{t:lower}, it remains to show that the same is true not just for
$\Ec^*$, but for any event $\Ec\subseteq\Ec^*$ with sufficiently high
probability. Suppose that $\Ec$ is such an event, with
$\delta=\Pr(\Ec\mid\Ec^*)\leq\Pr(\neg\Ec)\leq \frac1{4C\cdot16}(\frac
dm)^2$. Then, using $\tau_i=\frac m{\gamma
  n}[(\gamma\A^\top\S^\top\S\A)^{-1}]_{ii}$ as a shorthand, we have:
\begin{align*}
  \E[\tau_1\mid\Ec]
  &=
    \E[\tau_1\mid\Ec^*]
    +\frac{\delta}{1-\delta}\Big(
    \E[\tau_1\mid\Ec^*]
    -\E[\tau_1\mid\Ec^*,\neg\Ec]\Big)
    \geq \E[\tau_1\mid\Ec^*]-8\delta,
\end{align*}
where we used that $\delta\leq 1/2$ and, conditioned on $\Ec^*$, we have
$\tau_1\leq 4$. On the other hand,
\begin{align*}
  \E[\tau_2\mid\Ec]
  &=
    \E[\tau_2\mid\Ec^*]
    +\frac{\delta}{1-\delta}\Big(
    \E[\tau_2\mid\Ec^*]
    -\E[\tau_2\mid\Ec^*,\neg\Ec]\Big)
  \leq (1+2\delta)\,     \E[\tau_2\mid\Ec^*].
\end{align*}
Combining the two inequalities and using that
$\E[\tau_2\mid\Ec^*]\geq d/m$ and
$\delta\leq\frac1{4C\cdot16}(\frac dm)^2$, we get:
\begin{align*}
  \frac{\E\big[[(\gamma\A^\top\S^\top\S\A)^{-1}]_{11}\mid\Ec\big]}
  {\E\big[[(\gamma\A^\top\S^\top\S\A)^{-1}]_{22}\mid\Ec\big]}
  &=\frac{\E[\tau_1\mid\Ec]}{\E[\tau_2\mid\Ec]}
\geq \frac{(1+\frac d{4Cm})\E[\tau_2\mid\Ec^*] -
    8\delta}{(1+2\delta)\E[\tau_2\mid\Ec^*]}
  \\
  &\geq \frac{1+\frac{d}{4Cm} - 8\delta\frac md}{1+2\delta}\geq
    \frac{1+\frac d{8Cm}}{1+\frac d{32Cm}}\geq 1+\frac d{64Cm}.
\end{align*}
Thus, as discussed above, we conclude that for any scaling $\gamma>0$
and any event $\Ec$ with probability $\Pr(\Ec)\geq 1-\frac1{4C\cdot16}(\frac dm)^2$,
we have $\|\E[(\gamma\A^\top\S^\top\S\A)^{-1}\mid\Ec]-\I\|
=\Omega(\frac dm)$, which concludes the proof.

\subsection{Proof of Lemma \ref{l:inverse-moment}}
We conclude this section with a proof of the Binomial inverse moment
bound from Lemma \ref{l:inverse-moment}. While existing work has
focused on asymptotic expansions of inverse moments of the Binomial
\cite{znidaric2009asymptotic}, those precise 
characterizations either break down or appear to be impractical to work with
when the variable is significantly shifted, as in our case. Thus, we
use a different strategy: reducing the inverse moment bound to
showing an anti-concentration inequality for the Binomial
distribution. For this, we use the classical Paley-Zygmund inequality,
stated below.
\begin{lemma}\label{l:Paley-Zygmund}
  For any non-negative variable $\Z $ with finite variance and
  $\theta\in(0,1)$, we have:
  \begin{align*}
    \Pr\big(Z\geq \theta\,\E[Z]\big) \geq (1-\theta)^2\frac{\E[Z]^2}{\E[Z^2]}.
  \end{align*}
\end{lemma}

  Let $x\sim\binom(b,1/2)$ for a positive integer $b$. It follows that:
  \begin{align*}
      \E\Big[\frac1{x + b/2} - \frac1{b}\Big]
    &=
      \sum_{i=0}^{b}\Pr(x=i)\Big(\frac1{i+b/2}-\frac1b\Big) 
=\frac1b\sum_{i=0}^{b}\Pr(x=i)\frac{b/2-i}{b/2+i}
    \\
    &=\frac1b\sum_{i=0}^{\lfloor b/2\rfloor}\Pr(x=i) (b/2-i)\Big(\frac1{b/2+i} - \frac1{3b/2-i}\Big),
  \end{align*}
  where the last equality is obtained by symmetrically pairing up the
  terms $i$ and $b-i$ in the first sum. Next, observe that for $0\leq i\leq b/2-\sqrt b/4$, we have:
\begin{align*}
  (b/2-i)\Big(\frac1{b/2+i} - \frac1{3b/2-i}\Big)
  &\geq \frac{\sqrt
    b}{4}\Big(\frac1{b-\sqrt b/4} - \frac1{b+\sqrt b/4}\Big)
=\frac{\sqrt b}4\cdot\frac{\sqrt b/2}{b^2-b/16}
  \geq \frac1{8b}.
\end{align*}

Putting this together, we conclude that:
\begin{align}
   \E\Big[\frac1{x + b/2}\Big]\geq \Big(1+\frac1{8b}\Pr\big\{x-b/2\leq
  -\sqrt b/4\big\}\Big)\cdot \frac1b.\label{eq:anti}
\end{align}
Thus, it suffices to show that, with constant probability, $x$ is
smaller than its mean, $b/2$, by at least $\sqrt b/4$. This follows
from the Paley-Zygmund inequality (Lemma \ref{l:Paley-Zygmund}) by
setting $\Z =(x-b/2)^2$. Using standard formulas for the second and
fourth centered moment of the Binomial distribution, we have
$\E[Z]=b/4$ and $\E[Z^2]=\frac b4(1+\frac{3b-6}4)\leq 3b^2/16$. Therefore,
setting $\theta=1/4$ in Lemma \ref{l:Paley-Zygmund}, we obtain:
\begin{align*}
  \Pr\big(x-b/2\,\leq\, -\sqrt b/4\big)
  &=\frac12\,\Pr\big(|x-b/2|\,\geq\, \sqrt b/4\big)
=\frac12\,\Pr\big(Z\geq\theta\,\E[Z]\big)
  \\
  &\geq \frac12\,\Big(1-\frac14\Big)^2\,\frac{b^2/16}{3b^2/16}=\frac3{32}.
\end{align*}
Combining this with \eqref{eq:anti}, we obtain the desired claim for $C=8\cdot32/3$.

\section{Exact bias-correction for orthogonally invariant embeddings}
\label{s:orthogonal}

In this section we prove that orthogonal invariance implies no
inversion bias. This claim has been mentioned in the main text, in
Section \ref{noib}. Here we give a formal statement.

\begin{proposition}[Orthogonal invariance implies no inversion bias]\label{ro}
Let $\S$ be a random and right-orthogonally invariant matrix; specifically an $m \times n$ matrix (with $m \leq n$) such that for any orthogonal $n \times n$ matrix $\O$, we have $\S \stackrel{d}{=} \S\O$. Assume that $(\A^\top\S^\top\S\A)^{-1}$ exists with probability one.  Then the inversion bias is exactly correctable, i.e., there exists a constant $c = c_{m,n,d}$ such that $\E\bhS^{-1} = c\cdot \bS^{-1}$; where $\bS = \A^\top \A$ and $\bhS = \A^\top \S^\top \S \A$. 
\end{proposition}

Examples of orthogonal ensembles can be constructed in the following way: 
\begin{enumerate}
\item Let $\S$ have i.i.d.\@ normal entries with variance $m^{-1}$. Due to the properties of the Wishart ensemble, the constant $c_{m,n,d}$ is $c_{m,n,d} = m/(m-d-1)$.
\item
Let $\S_u$ be a uniformly random $m \times n$ partial orthogonal matrix (with $m \leq n$) such that $\S_u \S_u^\top = \I_m$. Equivalently, these are the first few rows of a Haar matrix. Then define $\S = \sqrt{n/m} \cdot \S_u$,  scaled such that $\E \S^\top \S  = \I_m$. We will call this the Haar sketch.
\item The class of orthogonally invariant matrices has several closure properties. Specifically, it is closed with respect to left-multiplication by any matrices, right-multiplication by orthogonal matrices, and with respect to vector space operations (addition and multiplication by scalars).  Several examples can be obtained this way. For instance, matrices $\S$ of the form $\S = \M\Z$, where $\Z $ has i.i.d.\@ normal entries with variance $m^{-1}$, and $\M $ is an arbitrary  matrix fixed or random and independent of $\Z $ are orthogonally invariant.
\end{enumerate}

\begin{proof}
We start with a reduction to orthogonal matrices: Let $\A = \U\bLambda \V^\top$ be the SVD of $\A$. Here recall that $\A$ is an $n \times d$ matrix, with $n \geq d$ and with full column rank, and thus $\U$ is an $n \times d$ partial orthogonal matrix with $n \geq d$, $\bLambda$ is $d \times d$ diagonal, and $\V$ is $d \times d$ orthogonal. Our goal is to show that $\E\bhS^{-1} = c\cdot \bS^{-1}$, or equivalently that

$$\E(\V\bLambda \U^\top \S^\top \S \U\bLambda \V^\top)^{-1} 
= c\cdot (\V\bLambda \U^\top \U\bLambda \V^\top)^{-1}.$$
Then, by cancelling $\bLambda$ and $\V$ above (using that they are deterministic square invertible matrices), and using that $\U^\top \U= \I$, we see that the above inequality is equivalent to 

$$
\E(\U^\top \S^\top \S \U)^{-1} = c\cdot \I.
$$
Thus, the problem is reduced to studying orthogonal matrices $\A =
\U$, such that $\bS= \A^\top \A = \U^\top \U = \I$. 
We claim that the right-orthogonal invariance implies that $\S \U  \stackrel{d}{=} \S\U\O$ for any $d \times d$ orthogonal matrix $\O $. Here is a geometric argument. We have that $\S \U $ are the angles that the random orthogonal rows of $\S $ form with the fixed set of basis vectors formed by the columns of $\U $. Also, $\S \U \O$ corresponds to the same quantity, but with respect to the basis formed by $\U \O$. Since $\S $ is right-rotationally invariant, these angles have the same distribution. 

Another, more algebraic proof is as follows. Since $\S $ is right-rotationally invariant, for any orthogonal $n \times n$ matrix $\bR $, we have $\S  \stackrel{d}{=} \S\bR$. Thus, for any fixed matrix $\U $, we have $\S \U  \stackrel{d}{=} \S\bR\U$. Choose a rotation matrix $\bR $ such that   $\bR \U\U^\top = \U \O\U^\top$, while $\bR \U^\perp$ is arbitrary, where $\U ^\perp$ is an orthogonal complement of $\U $. Then, multiplying the above with $\U^\top\U$ we have
$$\S \U \U^\top\U  \stackrel{d}{=} \S\bR\U \U^\top \U = \S\U \O\U^\top \U = \S\U\O.$$
We get that $\S \U  \stackrel{d}{=} \S\U\O$. 
Next, $\S \U  \stackrel{d}{=} \S\U\O$ implies that:
\begin{align*}
\U^\top \S^\top \S \U  &\stackrel{d}{=} \O^\top \U ^\top \S^\top \S \U \O\\
(\U^\top \S^\top \S  \U)^{-1} &\stackrel{d}{=} \O^\top  (\U^\top \S^\top \S  \U)^{-1} \O\\
\E (\U^\top \S^\top \S  \U)^{-1} &= \O^\top  \E(\U^\top \S^\top \S  \U)^{-1} \O.
\end{align*}
Thus, since $\J := \E(\U^\top \S^\top \S  \U)^{-1}$ is preserved under
conjugation by any orthogonal matrix, $\J$ must be a multiple of the
identity matrix, so $\J = c \I_d$, for some $c = c_{m,n,d}$. This
finishes the proof. 
\end{proof}

\end{document}